\documentclass[draftclsnofoot, onecolumn]{IEEEtran}
\pdfoutput=1

\usepackage{ifpdf}
\usepackage{cite}

\ifCLASSINFOpdf
  \usepackage[pdftex]{graphicx}
  \DeclareGraphicsExtensions{.pdf}
\else
  \usepackage[dvips]{graphicx}
  \DeclareGraphicsExtensions{.eps}
\fi

\usepackage{tikz}
\usetikzlibrary{arrows,shapes,backgrounds}

\usepackage{latexsym}
\usepackage[cmex10]{amsmath}
\usepackage{amssymb}
\usepackage{amsthm}

\usepackage{array}

\ifpdf
\usepackage[pdftex]{hyperref}
    \hypersetup{
    colorlinks=true,%
    citecolor=black,%
    filecolor=black,%
    linkcolor=black,%
    urlcolor=black
  }
\else
\usepackage[active]{srcltx}
\fi

\newtheorem{theorem}{Theorem}
\newtheorem{corollary}{Corollary}
\newtheorem{lemma}{Lemma}

\newtheorem{claim}{Claim}
\theoremstyle{definition}
\newtheorem{definition}{Definition}

\newcounter{example}
\newenvironment{example}%
{\refstepcounter{example}\par\vspace{10pt} \textit{Example \arabic{example}:}}%
{\hfill$\Diamond$}

\DeclareMathOperator{\E}{E}
\DeclareMathOperator{\Gr}{Gr}
\DeclareMathOperator{\Pj}{Pj}
\DeclareMathOperator{\rank}{rk}
\DeclareMathOperator{\Fr}{Fr}
\DeclareMathOperator{\entropy}{\mathcal{H}}
\DeclareMathOperator{\mutual}{\mathit{I}}
\DeclareMathOperator{\locrate}{\mathit{J}}

\newcommand{\etal}{\textit{et al.}}
\newcommand{\subs}{\text{SS}}

\newcommand{\ffield}{\mathbb{F}}

\newcommand{\lspan}[1]{\langle #1 \rangle}
\newcommand{\tr}{\top}

\newcommand{\gcos}[2]{\begin{bmatrix} #1 \\
#2 \end{bmatrix}}
\newcommand{\gco}[2]{\left[\substack{#1\\#2}\right]}
\newcommand{\cmat}[2]{\chi^{#1}_{#2}}
\newcommand{\cmatt}[2]{\zeta^{#1}_{#2}}

\newcommand{\bX}{\mathbf{X}}
\newcommand{\bY}{\mathbf{Y}}

\newcommand{\bH}{\mathbf{H}}

\newcommand{\loc}{\text{LOC}}

\begin{document}
\title{Capacity Analysis of Linear Operator Channels \\ over Finite Fields}

\author{Shenghao~Yang,~\IEEEmembership{Member,~IEEE},~Siu-Wai~Ho,~\IEEEmembership{Member,~IEEE},~Jin~Meng,
  and~En-hui~Yang~\IEEEmembership{Fellow,~IEEE}%
  \thanks{This paper was presented in part at IEEE Information Theory
    Workshop, Cairo, Egypt 2010, and at the IEEE International
    Symposium on Information Theory, Austin, USA 2010.}%
  \thanks{The work of S. Yang was supported in part
    by the National Basic Research Program of China Grant
    2011CBA00300, 2011CBA00301, the National Natural Science
    Foundation of China Grant 61033001, 61361136003.}%
  \thanks{The work of S.-W. Ho was supported by the Australian
    Research Council under an Australian Postdoctoral Fellowship as
    part of Discovery Project DP1094571.}%
  \thanks{The research of J. Meng and E.-h. Yang in this paper is
    supported in part by the Natural Sciences and Engineering Research
    Council of Canada under Grant RGPIN203035-11, and by the Canada
    Research Chairs Program.}%
  \thanks{S. Yang is with the Institute for Theoretical Computer
    Science, Institute for Interdisciplinary Information Sciences,
    Tsinghua University, Beijing, 100084, P. R. China. (e-mail:
    shyang@tsinghua.edu.cn)}%
  \thanks{S.-W. Ho is with the Institute for Telecommunications
    Research, University of South Australia, Australia. (e-mail:
    siuwai.ho@unisa.edu.au)} \thanks{J. Meng and E.-h. Yang are with
    the Department of Electrical and Computer Engineering, Waterloo
    University, Waterloo, ON, Canada. (e-mail:\{j4meng,
    ehyang\}@uwaterloo.ca)}%
}

\maketitle

\begin{abstract}
  Motivated by communication through a network employing linear
  network coding, capacities of linear operator channels (LOCs) with
  arbitrarily distributed transfer matrices over finite fields are
  studied.  Both the Shannon capacity $C$ and the subspace coding
  capacity $C_{\subs}$ are analyzed.  By establishing and comparing
  lower bounds on $C$ and upper bounds on $C_{\subs}$, various
  necessary conditions and sufficient conditions such that $C=
  C_\subs$ are obtained. A new class of LOCs such that $C=C_\subs$ is
  identified, which includes LOCs with uniform-given-rank transfer
  matrices as special cases. It is also demonstrated that $C_\subs$ is
  strictly less than $C$ for a broad class of LOCs.  In general, an
  optimal subspace coding scheme is difficult to find because it
  requires to solve the maximization of a non-concave
  function. However, for a LOC with a unique subspace degradation,
  $C_{\subs}$ can be obtained by solving a convex optimization problem
  over rank distribution. Classes of LOCs with a unique subspace
  degradation are characterized. Since LOCs with uniform-given-rank
  transfer matrices have unique subspace degradations, some existing
  results on LOCs with uniform-given-rank transfer matrices are
  explained from a more general way.
\end{abstract}


\section{Introduction}
Fix a finite field $\ffield$ with $q$ elements.  A linear operator
channel (LOC), also called a multiplicative matrix channel, with input
random variable $X\in\ffield^{T\times M}$ and output random variable
$Y\in\ffield^{T\times N}$ is given by
\begin{equation}\label{eq:formu}
  Y = XH,
\end{equation}
where $H\in \ffield^{M\times N}$ is called a \emph{transfer matrix}. We assume that $X$ and $H$ are
independent, and the transfer matrices in different channel uses are
independent and follow the same distribution.
For both the transmitter and  receiver,
the distribution of $H$ is given a priori, but the instances of $H$ are unknown.

A LOC is used to model communication through a network employing
linear network coding \cite{linear,alg}. Consider a network coding
scenario where the source node encodes its message into batches (also
called generations, classes or chunks), each of which contains $M$
packets of $T$ symbols \cite{chou03, maym06}. Intermediate network
nodes generate new packets by taking linear combinations of the
packages among the same batch. There may be packet loss and network
topological dynamics during the transmission.  The finally received
$N$ packets of a batch are all linear combinations of the original
packets of the batch. Such a network transmission can be modeled by a
LOC.

Coding problems for LOCs have been studied for various scenarios.  If
$T$ is much larger than $M$, parts of $X$ can be used to transmit an
identity matrix so that the receiver can recover the instances of
$H$. Such a scheme, called \emph{channel training}, has been widely
used for random linear network coding \cite{ho06j} and is
asymptotically optimal when $T$ goes to infinity.
The maximum achievable
rate of channel training (by multiple uses of the channel) can be
achieved using random linear codes \cite{yang10bf}, and a channel training
scheme with low encoding/decoding complexity has been proposed
\cite{yang11ac,yang12bats} by generalizing fountain codes. However, if $T$ is
not much larger than $M$, the overhead used to explicitly
recover the instances of $H$ is dominating, and hence different coding schemes must be studied.

We call the vector space spanned by the column vectors of a matrix $\bX$ the
column space of the matrix, denoted by $\lspan{\bX}$.
For a LOC, with probability one $\lspan{Y}$ is a subspace of $\lspan{X}$.
 Koetter and Kschischang
\cite{koetter08j} defined a channel with subspaces as input and output
to capture this property,
and discussed subspace codes for one use of this subspace channel.
They defined the minimum distance of a subspace code in terms of a
subspace distance between codewords, and used the minimum distance to
characterize the error (or erasure) correction capability of the
subspace code.  Thereafter, subspace coding has generated a lot of
research interests (see e.g., \cite{silva08j,silva09j,gadouleau10})
and the study of subspace coding has also been extended from one use to
multiple uses of the channel \cite{nobrega09,nobrega10}.

In this paper, we are interested in the achievable rates
of coding schemes when the error probability goes to zero asymptotically.  Most
existing works on subspace coding try to design large codebooks with
large minimum distances.  However, subspace codes designed under the
minimum distance criteria may not have a good performance for
multiple uses of a LOC
\cite{yang10bf}.

Towards better understanding of the coding problems and identifying new
directions to study coding for LOCs, an information theoretic study
of LOCs becomes necessary.  Existing works have studied several
classes of distributions of $H$.  When $M=N$, Silva \etal{} \cite{silva08c}
studied the case that $H$ is uniformly chosen from all full rank
$M\times M$ matrices. Siavoshani \etal{} \cite{siavoshani11} studied the
case that $H$ contains uniformly i.i.d. components.  N\'obrega \etal{}
\cite{nobrega11, nobrega11a} studied LOCs with
\emph{uniform-given-rank} transfer matrices, which include the
transfer matrices studied in \cite{silva08c, siavoshani11} as special
cases.
For all the above special distributions of $H$, it is shown
that $I(X;Y)=I(\lspan{X};\lspan{Y})$ for any input $X$,
which in turn implies that
using subspaces for encoding and decoding indeed achieves the Shannon capacity of these special LOCs; in addition, the
Shannon capacity of these LOCs  can be found by maximizing over 
input rank distribution.

However, many typical scenarios in linear network coding cannot be
covered by those special cases studied in the existing literature. Even
though the transfer matrix is full rank with high
probability for random linear network coding when both the field size
and the maximum flow from the source node to the destination node are
sufficiently large \cite{ho06j}, such a transfer matrix may not have the uniform
distribution studied in \cite{silva08c}. The transfer matrix studied in \cite{siavoshani11} can be
formed by using random linear network coding in the intermediate node
in Fig.~\ref{fig:three}, where node $a$ caches $M$ packets transmitted
by node $s$ before encoding, and transmits $N$ independent random linear
combinations of these $M$ packets. 
But if we take the packet loss during the transmission on both links
into consideration, the transfer matrix will not have independent
components since a packet loss will force a row/column to be zero.
Moreover, encoding after collecting $M$ packets 
introduces delay, so it is more practical to apply random linear network
coding in a causal way: the
intermediate node keeps transmitting the linear combinations of the packets it
has received~\cite{chou03,Lun2008}, which results in a transfer matrix
of the form (take $M=4$ as an example)
\begin{equation*}
  \begin{bmatrix}
    h_{1,1} & h_{1,2} & h_{1,3} & h_{1,4} & h_{1,5} & \cdots \\
    0 & h_{2,2} & h_{2,3} & h_{2,4} & h_{2,5} & \cdots \\
    0 & 0 & 0 & h_{3,4} & h_{3,5} & \cdots \\
    0 & 0 & 0 & 0 & h_{4,5} & \cdots
  \end{bmatrix},
\end{equation*}
where i) all nonzero rows are above any rows of all zeros, ii) 
the leading coefficient (the first nonzero component from the left) of
a nonzero row is not to the left of the leading coefficient of the row
above it, iii) all nonzero components are i.i.d. over a finite field. But such
a transfer matrix is even not uniform-given-rank. 
Furthermore, subspace coding is not capacity achieving in general.
 For example, when
$H$ is an $M\times M$ identity matrix and $T=1$, the Shannon capacity is
$M\log q$ bits per use and the subspace coding capacity is $1$ bit.

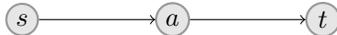
\begin{figure}
  \centering
  \begin{tikzpicture}[dot/.style={circle,draw=gray!80,fill=gray!20,thick,inner
 sep=2pt,minimum size=5pt}]
     \node[dot] (s) at(-2,0) {$s$};
     \node[dot] (a) at(0,0) {$a$} edge[<-] (s);
     \node[dot] (t) at(2,0) {$t$} edge[<-] (a);
  \end{tikzpicture}
  \caption{In this network, $s$ is the source node, $t$ is the destination node, and
    $a$ is the intermediate node that does not demand the file.}
  \label{fig:three}
\end{figure}

In this paper, we are motivated to study LOCs with arbitrarily
distributed transfer matrices.  We analyze both the Shannon capacity
and subspace coding capacity of LOCs, and we try to answer the
following questions: How to achieve or approach the Shannon capacity
of a LOC?  What is the performance of subspace coding and when is
subsapce coding optimal? How to
design subspace coding for general LOCs? Our results are for general values
of $T$, $M$, $N$ and $q$.

We first discuss some symmetry properties of LOCs, which lead to the
discovery that there exists a \emph{uniform-given-row-space} input
distribution achieving the Shannon capacity $C$ of a LOC (Theorem~\ref{the:diq}). We then derive an upper
bound and a lower bound on the Shannon capacity $C$, where the lower
bound is tight for  \emph{row-space-symmetric}
LOCs (Theorem~\ref{the:89s}) and is in general
at least as good as the lower bound derived using uniform-given-rank transfer
matrices in \cite{nobrega11, nobrega11a}.

We then turn our attention to the subspace coding capacity $C_{\subs}$
of a LOC. Note that a LOC has matrices as input and output, while
subspace coding uses subspaces for encoding and decoding. A general
way to study subspace coding for a LOC is to look at a subspace
degradation of the LOC, which is induced by a transition probability from
subspaces to matrices.  The subspace degradations induced by a LOC
are not unique in general, and finding an optimal subspace degradation
involves maximizing a non-concave function, which is in general
difficult to solve.  We study subspace coding with uniform-given-row-space input
distributions to obtain a lower bound  on the subspace
coding capacity (Theorem~\ref{the:diq2}), where the lower bound is further shown to be tight for LOCs with a unique
subspace degradation.  Optimal uniform-given-row-space input distributions for
subspace coding are characterized (Lemma~\ref{the:optssd} and
Theorem~\ref{the:mmm}), and the maximum achievable rate of
constant-rank uniform-given-row-space input distribution is given explicitly.  For a LOC with
a unique subspace degradation, the subspace coding capacity  can be
obtained by solving a convex optimization over the input rank
distribution (Theorem~\ref{the:uniquesd}), which generalizes the
similar result obtained for LOCs with uniform-given-rank transfer
matrices in \cite{nobrega11, nobrega11a}.
For row-space symmetric LOCs, an upper bound on $C_{\subs}$ is also
obtained (Lemma~\ref{lemma:rsuloc}).

To compare  $C_{\subs}$ with $C$,  we characterize, for both LOCs with a unique subspace degradation and row-space-symmetric LOCs,
 necessary conditions and sufficient conditions for $C_{\subs} = C$ (Theorem~\ref{the:capacitysd} and~\ref{the:non2}).
Subspace coding is not Shannon capacity achieving for both classes of LOCs if certain
Markov conditions are not satisfied.  On the other hand, 
subspace coding is capacity achieving for \emph{degraded}
LOCs, which has $I(X;Y)=I(\lspan{X};\lspan{Y})$ for
all input distributions.
A degraded LOC has a unique subspace degradation and is also row-space symmetric
(Theorem~\ref{the:degradedrowspace}).
The LOCs studied in
\cite{silva08c, siavoshani11, nobrega11, nobrega11a} are all
degraded.
We further characterize a new
class of degraded LOCs, called \emph{rank-symmetric} LOCs, and show
that a LOC with a uniform-given-rank transfer matrix is always rank
symmetric, but not vice versa when $T<M$
(Theorem~\ref{the:uniformeq}).

The relationship among the classes of LOCs characterized in this paper is
demonstrated in Fig.~\ref{fig:class}.  Note that when $T\geq M$, a
row-space-symmetric LOC always has a unique subspace degradation, but when
$T<M$, a row-space-symmetric LOC may not have a unique subspace degradation.

\begin{figure}
  \centering
  \begin{tikzpicture}
    \draw (-3.5,-1.6) rectangle (3.5,1.6) node[left=15pt,below=3pt] {$\Omega$};
    \draw (0:0.6cm) ellipse (2.5cm and 1.4cm) node[right=50pt] {$a$};
    \draw (180:0.6cm) ellipse (2.5cm and 1.4cm) node[left=50pt] {$b$};
    \draw (0,0.0) circle (1.0cm) node[above=10pt] {$c$};
    \draw (-90:0.2cm) circle (0.4cm) node {$d$};
  \end{tikzpicture}
  \caption{The Venn diagram about LOCs. $\Omega$ is the set of all
    LOCs. In additional to all LOCs, we
    study four subsets of LOCs: $a$ is the set of row-space-symmetric
    LOCs; $b$ is the set of LOCs with a
    unique subspace degradation; $c$ is the set of degraded
    LOCs; and $d$ is
    the set of rank-symmetric LOCs. Note that $a \subset b$ when
    $T\geq M$, and $d$ includes the LOCs studied in
\cite{silva08c, siavoshani11, nobrega11, nobrega11a}.}
  \label{fig:class}
\end{figure}
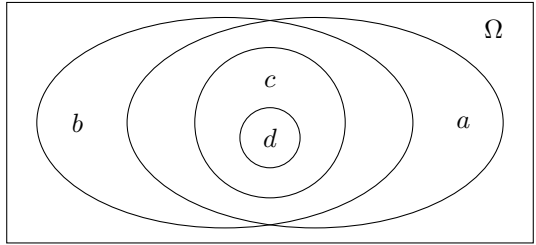

The rest of this paper is organized as follows.
After introducing some notations and mathematical results in
 Section~\ref{sec:pre}, we discuss symmetry properties of LOCs
 and bounds on $C$ in Section~\ref{sec:mul}.
Subspace coding for LOCs is studied in
Section~\ref{sec:subspace}. The comparison between $C$ and $C_{\subs}$
is made in Section~\ref{sec:subspace}. Finally, conclusion remarks are drawn in Section~\ref{sec:con}.

\section{Preliminaries}
\label{sec:pre}

Readers can skip
this section and come back later when these definitions/results are
referred to.

\subsection{Counting in Projective Space}

Let $\ffield$ be the finite field with $q$ elements.
Regard the vectors in $\ffield^t$ as column vectors.
For a matrix $\bX$, let $\rank(\bX)$ be the rank of $\bX$, let $\bX^\tr$ be the
transpose of $\bX$, and let $\lspan{\bX}$ be the subspace spanned by the
columns of $\bX$. We call $\lspan{\bX}$ and $\lspan{\bX^\tr}$ the
column space and the row space of $\bX$, respectively.

For a matrix $\mathbf B$
and a set of matrices $\mathcal A$, we define
\begin{equation*}
  \mathbf B + \mathcal{A} \triangleq \{\mathbf{B}+\mathbf{D} :\mathbf D \in \mathcal{A}\},
\end{equation*}
and
\begin{equation}\label{eq:msm}
  \mathbf B \mathcal{A} \triangleq \{\mathbf{BD}:\mathbf D \in \mathcal{A}\}.
\end{equation}
The multiplication $\mathcal{A}\mathbf B$ can be similarly defined.

The \emph{projective space} $\Pj(\ffield^t)$ is the collection of
all subspaces of $\ffield^t$. If $V$ is a subspace of $U$, we write
$V\leq U$. Define
\begin{equation*}
  \Pj(m,\ffield^t) \triangleq \{V:V\leq \ffield^t, \dim(V)\leq m\}.
\end{equation*}
This paper involves some counting results in projective spaces, some
of which
have been discussed in previous works (see
\cite{andrews76,gabidulin85,cooper00,gadouleau08,koetter08j,gadouleau10}
and the reference therein).
A self-contained discussion can be found in \cite{yang09v}.

Let $\Fr(\ffield^{m\times r})$ be the set of full rank matrices
in $\ffield^{m\times r}$. Define
\begin{equation}\label{eq:111}
\cmat{m}{r} \triangleq \left\{\begin{array}{ll}
    (q^m-1)(q^m-q)\cdots(q^m-q^{r-1}) &
    0 < r \leq m \\ 1 & r=0 \end{array} \right.
\end{equation}
For $r\leq m$, it is well-known that $|\Fr(\ffield^{m\times r})| =
\cmat{m}{r}$.
Define
\begin{equation}\label{eq:speaker}
  \cmatt{m}{r} \triangleq \cmat{m}{r} q^{-mr}.
\end{equation}
Since the number of $m\times r$ matrices is $q^{mr}$,
$\cmatt{m}{r}$ is equal to the probability that a
randomly chosen $m\times r$ matrix is full rank.

The \emph{Grassmannian}
$\Gr(r,\ffield^t)$ is the set of all $r$-dimensional
subspaces of $\ffield^t$. Thus $\Pj(m,\ffield^t) = \bigcup_{r\leq
  m}\Gr(r,\ffield^t)$.
The \emph{Gaussian binomial}  \cite{andrews76}
\begin{equation*} %
  \gcos{m}{r} \triangleq \frac{\cmat{m}{r}}{\cmat{r}{r}}
\end{equation*}
is the number of $r$-dimensional subspaces of $\ffield^m$,
i.e., $|\Gr(r,\ffield^m)|=\gco{m}{r}$.
Let
\begin{equation*} %
 \cmat{m,n}{r} \triangleq \frac{\cmat{m}{r}\cmat{n}{r}}{\cmat{r}{r}},
\end{equation*}
which is the number of $m\times n$ matrices with rank $r$ \cite{gabidulin85}.
So we have
\begin{equation}\label{eq:ckss66}
  \sum_{r} \cmat{m,n}{r} = q^{mn}.
\end{equation}

The following counting result is
a special case of \cite[Lemma 2]{gadouleau10}.

\begin{lemma}\label{lemma:c1}
  Let $V$ be an $s$-dimensional subspace of $\ffield^t$. For any
  integer $r$ with $s\leq r\leq t$,
    $$|\{U \in \Gr(r,\ffield^t): V\leq U \}| = \gcos{t-s}{r-s} = \gcos{t}{r}\frac{\cmat{r}{s}}{\cmat{t}{s}}.$$
\end{lemma}

\subsection{Probability Distribution over Matrices and Subspaces}

For a discrete random variable $X$, we use $p_X$ to denote its
probability mass function (PMF).  For two random variables $X$ and $Y$
defined on discrete alphabets $\mathcal{X}$ and $\mathcal{Y}$,
respectively, we write a transition probability (matrix) from
$\mathcal{X}$ to $\mathcal{Y}$ as $P_{Y|X}(\bY|\bX)$,
$\bX\in\mathcal{X}$ and $\bY\in \mathcal{Y}$.  
We say a transition matrix is \emph{deterministic} if all its entries
are either zero or one.
When it is clear from the context, we may omit the subscript of $p_X$ and $P_{Y|X}$ to simplify
  the notations.  Let $\entropy(X)$ be the entropy\footnote{The calligraphic $\mathcal{H}$ is used to denote entropy to make a distinction to the notion of the transfer matrix $H$.} of $X$ and
$\mutual(X;Y)$ be the mutual information between $X$ and $Y$.
We take logarithms to the base $2$.

For the sake of reference and comparison, we define three classes of conditionally uniform distributions that will be used in the paper. 

\begin{definition}[uniform-given-row-space distribution ($\alpha$-type
  distribution)] \label{def:uni1}
   A PMF $p$ over $\ffield^{m\times n}$ is
   \emph{uniform-given-row-space}
  if $p(\bX)=p(\bX')$ whenever $\lspan{\bX^\tr} =
  \lspan{\bX'^\tr}$.
  In other words, a random matrix $X\in \ffield^{m\times n}$ is uniform
  given row space if
  \begin{equation*}
    p_X(\bX) = \frac{p_{\lspan{X^\tr}}(\lspan{\bX^\tr})}{\cmat{m}{\rank(\bX)}}.
  \end{equation*}
\end{definition}

\begin{definition}[uniform-given-rank distribution] \label{def:uni2}
  A PMF $p$ over $\ffield^{m\times n}$ is \emph{uniform-given-rank} if
  $p(\bX) = p(\bX')$ whenever
  $\rank(\bX)=\rank(\bX')$. In other words, a random matrix $X\in \ffield^{m\times n}$ is 
  uniform-given-rank if
  \begin{equation*}
    p_X(\bX) = \frac{p_{\rank(X)}(\rank(\bX))}{\cmat{m,n}{\rank(\bX)}}.
  \end{equation*}
\end{definition}

\begin{definition}[uniform-given-dimension distribution]\label{def:ugd}
  A PMF $p$ over $\Pj(\ffield^T)$ is \emph{uniform-given-dimension} if
$p(V) =p(V')$ whenever $\dim(V)=\dim(V')$.
\end{definition}

A uniform-given-rank distribution is also uniform-given-row-space.
If $X$ has a uniform-given-row-space distribution,
$\lspan{X}$ has a uniform-given-dimension distribution.
Further define two classes of transition matrices with certain
symmetry properties as follows.

\begin{definition}[row-space-symmetric transition matrix]
  \label{def:rsu}
  A transition matrix $P(\cdot|\cdot):\ffield^{t\times
  m}\rightarrow \ffield^{t\times n}$ is said to be \emph{row-space-symmetric} if
  \begin{equation*}
    P(\bY|\bX) = P(\bY'|\bX')
  \end{equation*}
  whenever $\lspan{\bY}\leq \lspan{\bX}$,  $\lspan{\bY'}\leq
  \lspan{\bX'}$, $\lspan{\bX^\tr}=\lspan{\bX'^\tr}$ and $\lspan{\bY^\tr}=\lspan{\bY'^\tr}$.
  In other words, the transition probability $P(\bY|\bX)$, $\lspan{\bY}\leq \lspan{\bX}$, is determined by the row spaces of the input and output
matrices.
\end{definition}

\begin{definition}[rank-symmetric transition matrix]
  \label{def:rsu2}
  A transition matrix $P(\cdot|\cdot):\ffield^{t\times
  m}\rightarrow \ffield^{t\times n}$ is said to be \emph{rank-symmetric} if
  \begin{equation*}
    P(\bY|\bX) = P(\bY'|\bX')
  \end{equation*}
  whenever $\lspan{\bY}\leq \lspan{\bX}$,  $\lspan{\bY'}\leq
  \lspan{\bX'}$, $\rank(\bX)=\rank(\bX')$ and $\rank(\bY)=\rank(\bY')$.
  In other words, the transition probability $P(\bY|\bX)$, $\lspan{\bY}\leq \lspan{\bX}$, is determined by the ranks of the input and output
matrices.
\end{definition}

\section{Capacity of Linear Operator Channels}
\label{sec:mul}

A LOC defined in \eqref{eq:formu}, denoted by $\loc(H,T)$, is a
\emph{discrete memoryless channel} (DMC).
The dimensions of the transfer matrices discussed in this paper are
$M\times N$ unless otherwise specified.
Under the assumption that $H$ and $X$ are independent,
the transition probability $P_{Y|X}(\bY|\bX)$ is given by
\begin{equation*}
 P_{Y|X}(\bY|\bX)  = \Pr\{\bX H=\bY\}. %
\end{equation*}
The \emph{(Shannon) capacity} of
$\loc(H,T)$ is
\begin{equation*}%
 C = C(H,T)=\max_{p_X} \mutual(X;Y).
\end{equation*}
The input $X$, the output
$Y$, their row/column spaces and their ranks form
Markov chains shown in Fig.~\ref{fig:markov}.

 In this section,  we first introduce the essential technique
  of this paper---some
  symmetry properties of LOCs. We then investigate the input distributions that
  achieve the Shannon capacity, and give upper and lower bounds on the
  Shannon capacity.

\begin{figure}
  \centering
  \begin{tikzpicture}[scale=0.5]
    \node (X) at (-1,0) {$X$};
    \node (Y) at (1,0) {$Y$} edge[<-] (X);
    \node (Y1) at (3.6,1) {$\lspan{Y^\tr}$} edge[<-] (Y);
    \node (Y2) at (3.6,-1) {$\lspan{Y}$} edge[<-] (Y);
    \node (X1) at (-3.6,1) {$\lspan{X^\tr}$} edge[->] (X);
    \node (X2) at (-3.6,-1) {$\lspan{X}$} edge[->] (X);
    \node at (6.6,0) {$\rank(Y)$} edge[<-] (Y1) edge[<-]
    (Y2);
    \node at (-6.6,0) {$\rank(X)$} edge[->] (X1) edge[->] (X2);
  \end{tikzpicture}
  \caption{Random variables and Markov chains related to $\loc(H,T)$. All the random variables in a directed path form
a Markov chain.  For example, $\rank(X) \rightarrow \lspan{X}
\rightarrow X \rightarrow Y \rightarrow \lspan{Y} \rightarrow
\rank(Y)$ forms a Markov chain.}
  \label{fig:markov}
\end{figure}
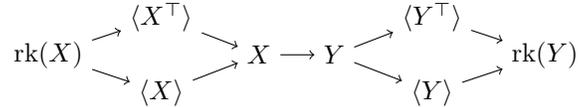

\subsection{Symmetry Properties}
\label{sec:symm}

The following lemma
demonstrates an intrinsic symmetry property of LOCs.
A matrix is said to have full column (row) rank if its rank is equal to
its number of columns (rows).

\begin{lemma}\label{the:symm}
  For $\loc(H,T)$, if $\bX = \mathbf{BD}$ and $\bY = \mathbf{BE}$
  where $\mathbf B$ has full column rank, then
  \begin{equation*}
    P_{Y|X}(\bY|\bX) = \Pr\{\bX H = \bY\} =
    \Pr\{\mathbf{D} H = \mathbf{E}\}.
  \end{equation*}
\end{lemma}
\begin{IEEEproof}
  The lemma follows from $P_{Y|X}(\bY|\bX) = \Pr\{\mathbf{BD} H =
  \mathbf{BE}\}  = \Pr\{\mathbf{D} H = \mathbf{E}\}$, where the
  last equality follows because $\mathbf B$ has full column rank.
\end{IEEEproof}

Recall that a DMC is defined to be \emph{symmetric} \cite{gallager} if
the set of outputs can be partitioned into subsets in such a way that
for each subset the matrix of transition probabilities (using inputs
as rows and outputs of the subset as columns) has the property that
each row is a permutation of each other row and each column (if more
than one) is a permutation of each other column.  The transition
matrix of a LOC satisfies properties similar to these of a symmetric
channel, but in general, a LOC is not a symmetric channel.

\begin{lemma}\label{prop:1}
  The transition matrix of $\loc(H,T)$ satisfies the
  following properties:
  \begin{enumerate}
  \item For $\bX_1,\bX_2\in \ffield^{T\times M}$
  with $\lspan{\bX_1^\tr} = \lspan{\bX_2^\tr}$ and $V\leq \ffield^N$, the vector
  $(P_{Y|X}(\bY|\bX_1):\bY\in \ffield^{T\times N}, \lspan{\bY^\tr}=V)$ is a permutation
  of the vector $(P_{Y|X}(\bY|\bX_2):\bY\in \ffield^{T\times N}, \lspan{\bY^\tr}=V)$;
  \item For $\bY_1,\bY_2\in \ffield^{T\times N}$
  with $\lspan{\bY_1^\tr} = \lspan{\bY_2^\tr}$ and $U\leq \ffield^M$, the vector
  $(P_{Y|X}(\bY_1|\bX):\bX\in \ffield^{T\times M}, \lspan{\bX^\tr}=U)$ is a permutation
  of the vector $(P_{Y|X}(\bY_2|\bX):\bX\in \ffield^{T\times M}, \lspan{\bX^\tr}=U)$.
  \end{enumerate}
\end{lemma}
\begin{IEEEproof}%
  Let $\phi(V) = \{\bY\in \ffield^{T\times N}, \lspan{\bY^\tr}=V\}$.
  To prove 1), we show that there exists a bijection $f:
  \phi(V)\rightarrow \phi(V)$ such that $\Pr\{\bX_1 H = \bY \} =
  \Pr\{\bX_2 H = f(\bY)\}$. Since $\lspan{\bX_1^\tr} =
  \lspan{\bX_2^\tr}$, there exists a full rank matrix $\mathbf T$ such
  that $\bX_2 = \mathbf T \bX_1$. Define $f:\phi(V)\rightarrow
  \phi(V)$ as $f(\bY) = \mathbf{T}\bY$. Since $\mathbf T$ is a full
  rank square matrix, $f$ is a bijection. The claim in 1) is verified
  by $\Pr\{\bX_2 H = f(\bY)\} = \Pr\{\bX_2 H = \mathbf{T}\bY\} =
  \Pr\{\mathbf T^{-1} \bX_2 H = \bY\} = \Pr\{\bX_1 H = \bY\}$, where
  the second equality follows from Lemma~\ref{the:symm}.

  The proof of 2) is similar and hence omitted.
\end{IEEEproof}

For the input matrices
with different row spaces, the rows of the transition matrix 
are usually not a permutation of each other.
The following result implied by Lemma~\ref{prop:1} will be used in
this paper.

\begin{lemma}\label{lemma:cond}
  For a LOC, if $\lspan{\bX^\tr} = \lspan{\bX'^\tr}$, then
  \begin{equation*}
    P_{\lspan{Y^\tr}|X}(V|\bX) = P_{\lspan{Y^\tr}|X}(V|\bX')
  \end{equation*}
  and
  \begin{equation*}
    P_{\rank(Y)|X}(s|\bX) = P_{\rank(Y)|X}(s|\bX').
  \end{equation*}
\end{lemma}
\begin{IEEEproof}
  Since $$P_{\lspan{Y^\tr}|X}(V|\bX)  = \sum_{\bY:\lspan{\bY^\tr}=V}
  P_{Y|X}(\bY|\bX),$$ the first equality follows from 1) in
  Lemma~\ref{prop:1}. The second equality follows from the first one.
\end{IEEEproof}

\subsection{Uniform-given-row-space Input Distributions}
\label{sec:alpha}

The intrinsic symmetry property of LOCs implies that the capacity-achieving
input distributions should have certain symmetry property, which is
characterized in the following theorem. Recall the definition of
uniform-given-row-space distribution in Definition~\ref{def:uni1}.

\begin{theorem}\label{the:diq}
  There exists a uniform-given-row-space input distribution that
  maximizes $\mutual(X;Y)$ for any LOC.
\end{theorem}
\begin{IEEEproof}%
  Let $p$ be an optimal input distribution for $\loc(H,T)$.
  For $\Phi\in \Fr(\ffield^{T\times T})$, define $p^{\Phi}$ as
  $p^{\Phi}(\bX) = p(\Phi\bX)$.
  First $p^{\Phi}$ is a PMF because $0\leq p^{\Phi}(\bX) =p(\Phi\bX) \leq 1$
  and $\sum_{\bX\in \ffield^{T\times M}} p^{\Phi}(\bX) = 1$.

  We show that
  $p^{\Phi}$ also achieves the capacity of the LOC. For the
  simplicity of the notations, we write $p' = p^{\Phi}$.
 Let $p_Y$ and $p_Y'$ be the PMF of $Y$ when the input distributions
 are $p$ and $p'$, respectively. We have
 \begin{IEEEeqnarray*}{rCl}
  p'_Y(\bY)
	& = &
	\sum_{\bX\in \ffield^{T\times M}} p'(\bX) P_{Y|X}(\bY|\bX) \\
	& = & \IEEEyesnumber \label{eq:osa-1}
	\sum_{\bX\in \ffield^{T\times M}}p(\Phi\bX) P_{Y|X}(\Phi\bY|\Phi\bX)  \\
	& = & \IEEEyesnumber \label{eq:osa-2}
	\sum_{\bX'\in \ffield^{T\times M}}p(\bX') P_{Y|X}(\Phi\bY|\bX') \\
	& = &
       p_Y(\Phi\bY),
 \end{IEEEeqnarray*}
 where \eqref{eq:osa-1} follows from Lemma~\ref{the:symm}
 and $p'(\bX)=p(\Phi\bX)$,
 and \eqref{eq:osa-2} follows by letting $\bX'=\Phi\bX$ and noting
 $\Phi\ffield^{T\times M} = \ffield^{T\times M}$. Therefore,
 \begin{IEEEeqnarray*}{rCl}
   \mutual(X;Y)|_{p'}
	& = &
	\sum_{\bX\in \ffield^{T\times M}}p'(\bX) \sum_{\bY\in \ffield^{T\times N}} P(\bY|\bX)
  \log  \frac{P(\bY|\bX)}{p'_Y(\bY)}  \\
	& = & %
	\sum_{\bX\in \ffield^{T\times M}}p(\Phi\bX) \sum_{\bY\in \ffield^{T\times N}} P(\Phi\bY|\Phi\bX) \log  \frac{P(\Phi\bY|\Phi\bX)}{p(\Phi\bY)} \\
	& = &
	\sum_{\bX'\in \ffield^{T\times M}}p(\bX') \sum_{\bY'\in \ffield^{T\times N}} P(\bY'|\bX') \log
        \frac{P(\bY'|\bX')}{p(\bY')} \\
	& = &
       \mutual(X;Y)|_{p},
 \end{IEEEeqnarray*}
 where the second equality follows from $p'_Y(\bY)=p_Y(\Phi\bY)$ and
 $P(\bY|\bX) = P(\Phi \bY|\Phi \bX)$ (cf. Lemma~\ref{the:symm}).

  Define $p^*$ as $$p^*(\bX) =
 \frac{1}{|\Fr(\ffield^{T\times T})|}\sum_{\Phi\in \Fr(\ffield^{T\times T})}
 p^{\Phi}(\bX).$$
 Since mutual information is a concave function of the input
 distribution \cite{gallager},
 \begin{IEEEeqnarray*}{rCl}
    \mutual(X;Y)|_{p^*} & \geq & \frac{1}{|\Fr(\ffield^{T\times T})|}\sum_{\Phi\in \Fr(\ffield^{T\times T})}
    \mutual(X;Y)|_{p^{\Phi}} \\
    & = & C(H,T).
 \end{IEEEeqnarray*}
Thus, $p^*$ is also an optimal input distribution for the channel.
The proof is completed by noting that $p^*$ is uniform-given-row-space.
\end{IEEEproof}

Theorem~\ref{the:diq} reveals that 
uniform-given-row-space input distributions can match the intrinsic
symmetry of LOCs.
We will show more 
applications of uniform-given-row-space input distributions in this paper.

In the remaining part of this subsection, we discuss how the
calculation of the transition matrix and the channel capacity can be
simplified by 
the symmetry properties and uniform-given-row-space input distributions. 
To compute the channel
capacity of $\loc(H,T)$, the first step is to
compute the matrix of transition probabilities using the distribution
of $H$. 
A straightforward computation of the transition matrix from the
distribution of $H$ requires the calculation of
  $q^{T(M+N)}$ components of the transition matrix. 
But using the symmetry properties, this number can be reduced to
\begin{equation*}
  \sum_{k=0}^{\min\{T,M\}} \gcos{M}{k} q^{kN}
  < \left\{\begin{array}{ll}
      cq^{MN} & M\leq \min\{T,N\} \\
      c'q^{L(M+N-L)} & \text{otherwise},
    \end{array} \right.
\end{equation*}
where $L=\min\{T,(M+N)/2\}$, $c$ and $c'$ are constants (ref. Appendix~\ref{sec:symmopt}).

The input distribution of a LOC has $q^{TM}$ probability masses. To
find an optimal input distribution, a straightforward approach needs
to determine $q^{TM}-1$ out of them. 
Theorem~\ref{the:diq} enable us to focus on uniform-given-row-space
input distributions, which is determined by
 a PMF over $\Pj(\min\{M,T\},\ffield^{M})$.
Thus the number of probability masses to determine can be
reduced to
\begin{equation*}
  \sum_{k=0}^{\min\{M,T\}} \gcos{M}{k} <
  \left\{
      \begin{array}{ll}
        \Theta_1q^{M^2/4} & \text{for}\ T\geq M/2 \\
        \Theta_2q^{T(M-T)} & \text{otherwise},
      \end{array}
   \right.
\end{equation*}
where $\Theta_1$ and $\Theta_2$ are constants (ref. Appendix~\ref{sec:symmopt}).
Note that those computations are still complicated for relatively
large $M$ and $T$.

\subsection{Upper and Lower Bounds on $C$}
\label{sec:boundsss}

We derive bounds on
$\mutual(X;Y)$ with the addition of two terms: one corresponds to the intrinsic
symmetry and another one is the achievable rate of the
channel given by $P_{\lspan{Y^\tr}|\lspan{X^\tr}}$.
Note that in Lemma~\ref{prop:1}, the symmetry property only holds for
input (output) matrices sharing the same row space.  Roughly, the
transition matrix $P_{\lspan{Y^\tr}|\lspan{X^\tr}}$ captures some
property of a LOC that is opposite to the intrinsic symmetry.

The transition matrix
$P_{\lspan{Y^\tr}|\lspan{X^\tr}}$ is solely determined by $p_H$. 
By Lemma~\ref{lemma:cond}, $P_{\lspan{Y^\tr}|\lspan{X^\tr}}(V|U) =
P_{\lspan{Y^\tr}|X}(V|\bX)$ for any $\bX$ with $\lspan{\bX^\tr}=U$. 
Let
\begin{equation}\label{eq:ratej}
    \locrate(\rank(X);\rank(Y)) \triangleq \sum_{s,r}
  p_{\rank(X)\rank(Y)}(r,s) \log\frac{\cmat{T}{s}}{\cmat{r}{s}}.
\end{equation}
Note that $\locrate(\rank(X);\rank(Y))$ is always nonnegative (cf. the
definition of $\cmat{r}{s}$ in \eqref{eq:111}), and
$p_{\rank(X)\rank(Y)}(r,s)$ can be solely derived from
$p_{\lspan{X^\tr}}$ and $P_{\lspan{Y^\tr}|\lspan{X^\tr}}$  as
\begin{equation*}%
  p_{\rank(X)\rank(Y)}(r,s)
  =
  \sum_{U\in \Gr(r,\ffield^M)} P_{\rank(Y)|\lspan{X^\tr}}(s|U) p_{\lspan{X^\tr}}(U),
\end{equation*}
where $P_{\rank(Y)|\lspan{X^\tr}}(s|U) = \sum_{V:\dim(V)=s}P_{\lspan{Y^\tr}|\lspan{X^\tr}}(V|U)$.

\begin{theorem}\label{the:89s}
Consider $\loc(H,T)$ with input $X$ and output $Y$.
  For a uniform-given-row-space input distribution,
  \begin{equation}\label{eq:lower11}
    \mutual(X;Y) \geq \locrate(\rank(X);\rank(Y))  + \mutual(\lspan{X^{\tr}}; \lspan{Y^{\tr}})
  \end{equation}
  and
  \begin{IEEEeqnarray*}{rCl}
    \mutual(X;Y) & \leq & \locrate(\rank(X);\rank(Y))  +
    \mutual(\lspan{X^{\tr}}; \lspan{Y^{\tr}}) \\
    & & \IEEEyesnumber \label{eq:upper11}
    + \sum_{s,r}
  p_{\rank(X)\rank(Y)}(r,s) \log {\cmat{r}{s}},
  \end{IEEEeqnarray*}
  where the equality in \eqref{eq:lower11} holds when the LOC has a
  row-space-symmetric transition matrix.
\end{theorem}
\begin{IEEEproof}
  Fix a uniform-given-row-space input distribution $p_X$.
  Let $Y^*$ be a random matrix over $\ffield^{T\times N}$ with transition probability
  \begin{equation*}
  P_{Y^*|X}(\bY|\bX) = \left \{
  \begin{array}{ll}
    \frac{P_{\lspan{Y^\tr}|\lspan{X^\tr}}(\lspan{\bY^\tr}|\lspan{\bX^\tr})}{\cmat{\rank(\bX)}{\rank(\bY)}}
    & \lspan{\bY}\leq \lspan{\bX}, \\
    0 & \text{otherwise}.
  \end{array}
  \right.
  \end{equation*}
  Note that $P_{Y^*|X}$ is row-space symmetric.

  The proofs of the following  claims are given in Appendix~\ref{sec:pf-claims}.

  \begin{claim}\label{claim:1}
    For a uniform-given-row-space  input distribution $p_X$,
    $\mutual(X;Y) \geq \mutual(X;Y^*)$, with equality when $P_{Y^*|X}=P_{Y|X}$.
  \end{claim}

  We can show Claim~\ref{claim:1} using the property that for fixed
  $p_X$, mutual information $I(X;Y)$ is a convex function of the
  transition probabilities. 
  We can further show the following claim by directly applying  the definition of $P_{Y^*|X}$.

  \begin{claim}\label{claim:2}
    For a uniform-given-row-space  input distribution $p_X$,
    \begin{equation} \label{claim2eq1}%
      \entropy(Y^*|X) = \sum_{s\leq r} p_{\rank(X)\rank(Y)}(r,s) \log {\cmat{r}{s}} + \entropy(\lspan{Y^\tr}|\lspan{X^\tr})
    \end{equation}
    and
    \begin{equation} \label{claim2eq2}%
      \entropy(Y^*) = \sum_s p_{\rank(Y)}(s) \log \cmat{T}{s} +
  \entropy(\lspan{Y^\tr}).
    \end{equation}
  \end{claim}

By Claim~\ref{claim:2},
\begin{IEEEeqnarray*}{rCl}
  I(X;Y^*)
  & = &
  \entropy(Y^*) - \entropy(Y^*|X) \\
  & = &
  \sum_{s\leq r} p_{\rank(X)\rank(Y)}(r,s) \log \frac{\cmat{T}{s}}{\cmat{r}{s}} + I(\lspan{X^\tr};\lspan{Y^\tr}),
\end{IEEEeqnarray*}
which, together with Claim~\ref{claim:1}, proves \eqref{eq:lower11}.

To prove \eqref{eq:upper11}, we  have
\begin{IEEEeqnarray*}{rCl}
  \entropy(Y)
  & = &
  \sum_s \sum_{V\in \Gr(s,\ffield^N)} \sum_{\bY:\lspan{\bY^\tr}=V}p_{Y}(\bY)\log \frac{1}{p_Y(\bY)} \\
  & \leq & \IEEEyesnumber \label{eq:90wk-1}
  \sum_s \sum_{V\in \Gr(s,\ffield^N)} p_{\lspan{Y^\tr}}(V) \log \frac{\cmat{T}{s}}{p_{\lspan{Y^\tr}}(V)} \\
  & = & \IEEEyesnumber \label{eq:90wk-2}
  \entropy(Y^*),
\end{IEEEeqnarray*}
where \eqref{eq:90wk-1} is derived by the log-sum inequality (cf. \cite{yeung08}) and
\eqref{eq:90wk-2} is obtained by \eqref{claim2eq2}. Then,
\begin{IEEEeqnarray*}{rCl}
  I(X;Y)-I(X;Y^*)
  & = &
  \entropy(Y)-\entropy(Y|X) - \entropy(Y^*) + \entropy(Y^*|X) \\
  & \leq &
  \entropy(Y^*|X) - \entropy(Y|X) \\
  & \leq &
  \sum_{s\leq r} p_{\rank(X)\rank(Y)}(r,s) \log {\cmat{r}{s}},
\end{IEEEeqnarray*}
where the last inequality follows from \eqref{claim2eq1} and $\entropy(Y|X) \geq \entropy(\lspan{Y^\tr}|X) =
\entropy(\lspan{Y^\tr}|\lspan{X^\tr})$ since
$p_{\lspan{Y^\tr}|X}(V|\bX)$ depends on $\bX$ only through
$\lspan{\bX^\tr}$ (cf. Lemma~\ref{lemma:cond}).
\end{IEEEproof}

The lower bound in the above theorem can be achieved by a coding
scheme employing a superposition structure, which includes a cloud
code and a set of satellite codes, each of which corresponds to a
cloud center. The rate $I(\lspan{X^\tr};\lspan{Y^\tr})$ can be
achieved by the cloud code, while the rate
$\locrate(\rank(X);\rank(Y))$ can be achieved by the satellite
codes. Readers are referred to \cite{yang12sumas} for detailed
discussion of this coding scheme.

Our lower bound is at least as good as the lower bound
obtained in \cite{nobrega11, nobrega11a}, 
where a transfer matrix is converted to a uniform-given-rank
transfer matrix.
We will compare these two bounds at the end of Section~\ref{sec:34}.

In the following sections, we will see that the quantity
$\locrate(\rank(X);\rank(Y))$ is also
related to the coding rate of subspace coding.
In the definition of $\locrate(\rank(X);\rank(Y))$,
the inverse of the term $\frac{\cmat{T}{s}}{\cmat{r}{s}}$ has the following meaning.
Let $V$ be an $s$-dimensional subspaces of $\ffield^T$, which can be
regarded as the column space of the output matrix.
The number of $r$-dimensional
subspaces of $\ffield^T$ is $\gco{T}{r}$; and by Lemma~\ref{lemma:c1}, the number of $r$-dimensional subspaces of $\ffield^T$
that include $V$, which are the possible column spaces of the
input matrix, is $\gco{T}{r}\frac{\cmat{r}{s}}{\cmat{T}{s}}$.
Thus, the fraction of the number of $r$-dimensional subspaces of $\ffield^T$
that include $V$ as a subspace is exactly
$\frac{\cmat{r}{s}}{\cmat{T}{s}}$.

Let us look at
another property of the quantity $\locrate(\rank(X);\rank(Y))$.
Let
\begin{equation*}
\epsilon(T,q)\triangleq \sum_{s}p_{\rank(H)}(s)
\log\frac{\cmatt{T}{s}}{\cmatt{M}{s}},
\end{equation*}
where $\cmatt{m}{r}$ is defined in \eqref{eq:speaker}.
\begin{lemma}\label{lemma:c5}
If $T\geq M$ and $p_{\rank(X)}(M)=1$,
\begin{IEEEeqnarray*}{rCl}
  \locrate(\rank(X);\rank(Y))
  & = & %
  (T-M)\E[\rank(H)] \log q + \epsilon(T,q),
\end{IEEEeqnarray*}
where $0\leq \epsilon(T,q)< 1.8$ for all $T$ and $q$.
\end{lemma}
\begin{IEEEproof}
When $T\geq M$ and $p_{\rank(X)}(M)=1$,
\begin{IEEEeqnarray*}{rCl}
  \locrate(\rank(X);\rank(Y))
  & = & \IEEEyesnumber \label{eq:ci-1}
  \sum_{s}p_{\rank(H)}(s) \log q^{(T-M)s}\frac{\cmatt{T}{s}}{\cmatt{M}{s}} \\
  & = &
  (T-M)\E[\rank(H)] \log q + \epsilon(T,q),
\end{IEEEeqnarray*}
where \eqref{eq:ci-1} follows that $\rank(H)=\rank(Y)$ since $X$ has
full column rank.

The lower bound on $\epsilon(T,q)$ holds due to $T\geq M$, and the
upper bound on $\epsilon(T,q)$ is obtained by bounding $\cmatt{m}{r}$
using a constant given in \cite{cooper00}.
\end{IEEEproof}

The above lemma tells us that when $T>M$,
 $\locrate(\rank(X);\rank(Y))$ is larger than $(T-M)\E[\rank(H)]
\log q$, which is the maximum achievable rate of channel
training \cite{yang10bf}.  (Recall that in channel training, $M$ rows of
$X$ are used to recover the transfer matrix in the receiver.)
We know that subspace coding can in general do better than channel
training \cite{koetter08j}.
The lower bound in Theorem~\ref{the:89s} implies
that the rate gain is at least
\begin{equation*}
  \max_{p_X:\text{uniform-given-row-space},p_{\rank(X)}(M)=1}\epsilon(T,q).
\end{equation*}
(Note that $\mutual(\lspan{X^{\tr}}; \lspan{Y^{\tr}})=0$ when
$p_{\rank(X)}(M)=1$ since $\lspan{X^{\tr}}=\ffield^M$ is deterministic when $\rank(X)=M$.)

\begin{example}\label{ex:j}
Consider $\loc(H,1)$, for which channel training is not useful. We have
\begin{IEEEeqnarray*}{rCl}
  \locrate(\rank(X);\rank(Y))
  & = &
  \sum_{r\in\{0,1\}} \sum_{s\in
    \{0,1\}:s\leq r} p_{\rank(X)\rank(Y)}(r,s) \log
  \frac{\cmat{1}{s}}{\cmat{r}{s}} \\
  & = &
  p_{\rank(X)\rank(Y)}(0,0) \log \frac{\cmat{1}{0}}{\cmat{0}{0}}
  + p_{\rank(X)\rank(Y)}(1,0) \log \frac{\cmat{1}{0}}{\cmat{1}{0}} \\
  & &
  + p_{\rank(X)\rank(Y)}(1,1) \log \frac{\cmat{1}{1}}{\cmat{1}{1}} \\
  & = & 0.
\end{IEEEeqnarray*}
The lower bound in Theorem~\ref{the:89s} for $\loc(H,1)$ becomes 
$\mutual(\lspan{X^{\tr}}; \lspan{Y^{\tr}})$, and the gap between the
upper bound and the lower bound is $p_{\rank(X)\rank(Y)}(1,1)\log
(q-1)$. Note that $\mutual(\lspan{X^{\tr}}; \lspan{Y^{\tr}})$ can be
large. For example, when $H$ is the $M\times M$ identity matrix,
$\mutual(\lspan{X^{\tr}}; \lspan{Y^{\tr}}) = \log |\Gr(1,\ffield^M)| = \log
\frac{q^M-1}{q-1} \geq (M-1)\log q$. 
\end{example}

\subsection{Properties of Row-Space-Symmetric LOCs}

We call a LOC \emph{row-space-symmetric} if its transition matrix is
row-space-symmetric. The lower bound in Theorem~\ref{the:89s} is tight
for row-space-symmetric LOCs. We introduce some properties of such
LOCs to be used in other sections.

By definition, a LOC is row-space symmetric if and only if
for any $\bX$ and $\bY$ with $\lspan{\bY}\leq \lspan{\bX}$,
  \begin{equation*}
    P_{Y|X}(\bY|\bX) = \frac{1}{\cmat{\rank(\bX)}{\rank(\bY)}} P_{\lspan{Y^\tr}|\lspan{X^\tr}}(\lspan{\bY^\tr}|\lspan{\bX^\tr}),
  \end{equation*}
where $\cmat{\rank(\bX)}{\rank(\bY)}$ is the number of $\bY_1$ such
that $\lspan{\bY_1}\leq \lspan{\bX}$ and $\lspan{\bY_1^\tr} = \lspan{\bY^\tr}$, and $P_{\lspan{Y^\tr}|\lspan{X^\tr}}(\lspan{\bY^\tr}|\lspan{\bX^\tr})$ is only determined by $p_H$.

\begin{lemma}\label{lemma:new1}
When $T\geq M$, $\loc(H,T)$ being row-space-symmetric implies that $H$ is
uniform-given-row-space.
\end{lemma}
\begin{IEEEproof}
Let $\bX$ and $\bX'$ be two full-rank input matrices with
  $\lspan{\bX^\tr} = \lspan{\bX'^{\tr}}$. Let $\bH$ and $\bH'$ be two transfer matrices with
  $\lspan{\bH^\tr} = \lspan{\bH'^{\tr}}$. Since $T\geq M$, we have
  $\lspan{(\bX\bH)^\tr} = \lspan{\bH^\tr}$ and $\lspan{(\bX\bH')^\tr}
  = \lspan{\bH'^\tr}$. Hence $\lspan{(\bX\bH)^\tr} =
  \lspan{(\bX\bH')^\tr}$. By the definition of row-space-symmetric
  LOCs, we have $p_{Y|X}(\bX\bH|\bX) = p_{Y|X}(\bX'\bH'|\bX')$, which
  implies $p_H(\mathbf{H}) = p_H(\mathbf{H}')$. 
\end{IEEEproof}

When $T<M$, it is not
necessary that the transfer matrix of a  row-space-symmetric LOC
satisfies the above constraint. 

\begin{example}\label{ex:rss}
We denote a LOC with $T=1$ over the binary
field $\ffield_2$ as
$\loc_{2}(H,1)$, where $H\in\ffield_2^{M\times N}$ is the transfer matrix.
The input and the output of $\loc_{2}(H,1)$ are in the same set
$\ffield_2^{1\times M}$.
Since the mapping from $X$ to $\lspan{X^T}$ in this special case is a
bijection, we have $P_{Y|X}(\bY|\bX) =
P_{\lspan{Y^\tr}|\lspan{X^\tr}}(\lspan{\bY^\tr}|\lspan{\bX^\tr})$ for
any $\bX,\bY\in \ffield_2^{1\times M}$ with $\lspan{\bY}\leq
\lspan{\bX}$. Hence, $\loc_{2}(H,1)$ is row-space-symmetric for any
distribution of $H$. 
\end{example}

\section{Subspace Coding Capacity of LOCs}
\label{sec:subspace}

One of the intrinsic properties of LOCs is that $\lspan{Y} \leq
\lspan{X}$. If we restrict to the column spaces of
the input and output of a LOC, it is possible that simpler encoding/decoding schemes can
be developed. This approach, called subspace coding, was first
adopted by Koetter and Kschischang \cite{koetter08j} for random linear
network coding.
In this section, we characterize the asymptotic performance of
subspace coding with multiple uses of the channel
when the error probability goes to zero.

We discuss how to characterize the maximum achievable rate of subspace
coding (also known as the \emph{subspace coding capacity}) and provide
lower bounds on the subspace coding capacity. We also introduce an important 
class of LOCs, for which the optimal subspace coding scheme is relatively
easier to find.

\subsection{Optimal Subspace Degradations}
\label{sec:osd}

A LOC is a matrix channel, so it must be converted to a subspace
channel to use subspace coding. An $n$-block subspace code is a subset
of $(\Pj(\min\{T,M\},\ffield^T))^n$. To apply a subspace code to a
LOC, the subspaces in a codeword need to be converted to matrices. For
$U\in \Pj(\min\{T,M\},\ffield^T)$, this conversion can be done by a
transition probability $P_{X|\lspan{X}}(\cdot|U)$.  The decoding of a
subspace code also uses only the column spaces spanned by the received
matrices.  Given a transition matrix $P_{X|\lspan{X}}$, we have a new
channel with input $\lspan{X}$ and output $\lspan{Y}$.

\begin{definition}\label{def:deg}
  For $\loc(H,T)$ with a given a transition probability $P_{X|\lspan{X}}$,
  we have a new channel law given by
  \begin{equation}
    \label{eq:decom}
    P_{\lspan{Y}|\lspan{X}}(V|U) = \sum_{\bX}P_{\lspan{Y}|X}(V|\bX) P_{X|\lspan{X}}(\bX|U).
  \end{equation}
  This channel takes subspaces as input and output and is called a
  \emph{subspace degradation} of $\loc(H,T)$ with respect to
  $P_{X|\lspan{X}}$.
\end{definition}

The capacity of the subspace degradation of $\loc(H,t)$ w.r.t.
$P_{X|\lspan{X}}$ is $\max_{p_{\lspan{X}}} \mutual(\lspan{Y};
\lspan{X})$.  Therefore, the {subspace coding capacity} of $\loc(H,T)$ is
\begin{IEEEeqnarray}{rCl}
  C_\subs = C_{\subs}(H,T)
  & \triangleq & \IEEEnonumber
  \max_{P_{X|\lspan{X}}}\max_{p_{\lspan{X}}}\mutual(\lspan{X};\lspan{Y})\\
  & = &  \label{eq:cx-1}
  \max_{p_X}\mutual(\lspan{X};\lspan{Y}).
\end{IEEEeqnarray}
To verify \eqref{eq:cx-1}, we see that for given $P_{X|\lspan{X}}$ and
$p_{\lspan{X}}$, the PMF of $X$ is given by
$p_X(\bX)=p_{\lspan{X}}(\lspan{\bX})P_{X|\lspan{X}}(\bX|\lspan{\bX})$.
On the other hand, fix a distribution $p_X$. The distribution
$p_{\lspan{X}}$ can be derived, and the distribution
$P_{X|\lspan{X}}(\cdot|U)$ can be derived for any $U$ with $p_{\lspan{X}}(U)\neq 0$. If $p_{\lspan{X}}(U)= 0$, the
distribution $P_{X|\lspan{X}}(\cdot|U)$ does not appear in the maximization of
$\mutual(\lspan{X};\lspan{Y})$.

When $P_{X|\lspan{X}}$ is fixed, $P_{\lspan{Y}|\lspan{X}}(V|U)$ is
also fixed (cf. \eqref{eq:decom}), and hence
$\mutual(\lspan{X};\lspan{Y})$ is a concave function of
$p_{\lspan{X}}$. On the other hand, when $p_{\lspan{X}}$ is fixed,
$P_{\lspan{Y}|\lspan{X}}(V|U)$ is a linear function of
$P_{X|\lspan{X}}$ (cf. \eqref{eq:decom}) and
$\mutual(\lspan{X};\lspan{Y})$ is a convex function of
$P_{\lspan{Y}|\lspan{X}}(V|U)$, and hence
$\mutual(\lspan{X};\lspan{Y})$ is a convex function of
$P_{X|\lspan{X}}$.  (We can similarly argue that
$\mutual(\lspan{X};\lspan{Y})$ is not concave in $p_X$ in general.)  Hence, finding
an optimal subspace coding scheme involves maximizing a non-concave
function, which is in general a difficult problem due to computational
complexity.

Recall that a transition matrix is \emph{deterministic} if all its entries
are either zero or one.
We can simplify the problem of finding an optimal subspace degradation
by considering only deterministic transition matrices.

\begin{lemma}\label{the:beta}
There exists an optimal subspace degradation w.r.t. a deterministic
transition matrix $P_{X|\lspan{X}}$.
\end{lemma}
\begin{IEEEproof}
  Consider a procedure as follows.  Fix $p_{\lspan{X}}$ and
  $P^0_{X|\lspan{X}}$ that achieve $C_{\subs}(H,T)$. If
  $P^0_{X|\lspan{X}}$ is deterministic, the procedure stops.
  Otherwise, there must exist $U \in \Gr(\min\{T,M\},\ffield^T)$ such
  that $P^0_{X|\lspan{X}}(\bX|U) < 1$ for all input $\bX$ with
  $\lspan{\bX} = U$.

  For each $\bX$ with $\lspan{\bX} = U$, define
  $P^{\bX}_{X|\lspan{X}}$ as
  $P^{\bX}_{X|\lspan{X}}(\cdot|U')=P^{0}_{X|\lspan{X}}(\cdot|U')$ for
  $U'\neq U$ and $P^{\bX}_{X|\lspan{X}}(\bX|U)=1$.  We can write
  \begin{equation*}
    P^0_{X|\lspan{X}}(\cdot|\cdot) = \frac{1}{\cmat{M}{\dim(U)}}
      \sum_{\bX:\lspan{\bX}=U} P^0_{X|\lspan{X}}(\bX|U) P^{\bX}_{X|\lspan{X}}(\cdot|\cdot).
  \end{equation*}
  Since $\mutual(\lspan{X};\lspan{Y})$ is a convex function of
  $P_{X|\lspan{X}}$, there exists $\bX_0$ with $\lspan{\bX_0}=U$ such that
  \begin{equation*}
    \mutual(\lspan{X};\lspan{Y})\big |_{P^{\bX_0}_{X|\lspan{X}}} \geq \mutual(\lspan{X};\lspan{Y})\big |_{P^0_{X|\lspan{X}}}.
  \end{equation*}
  Hence the subspace degradation associated with
  $P^{\bX_0}_{X|\lspan{X}}$ is also optimal.
  We then repeat the above procedure with $P^{\bX_0}_{X|\lspan{X}}$ in
  place of $P^0_{X|\lspan{X}}$.

  The above procedure must stop in finite steps since $\Gr(\min\{T,M\},\ffield^T)$ has
  finite elements.  $P^0_{X|\lspan{X}}$ in the final step is
  deterministic. 
\end{IEEEproof}

Lemma~\ref{the:beta} 
  enables us to focus on a finite set of deterministic 
  transition matrices $P_{X|\lspan{X}}$ to find the optimal subspace degradation.
For small $T$, it is possible to numerically evaluate all the
deterministic transition matrices $P_{X|\lspan{X}}$.

\begin{example} \label{example:ssopt}
  We use $\loc(H,1)$ as an
  example to show how to evaluate the subspace coding capacity.
The input and  output of a subspace
degradation can be two subspaces $\lspan{0} \triangleq \{0\}$
and $\lspan{1} \triangleq \{0, 1\}$.
By Lemma~\ref{the:beta}, we only need to consider subspace
degradations with
$P_{X|\lspan{X}}(\bX|\lspan{1})=1$ for certain $\bX \in \ffield^{1\times
  M}\setminus \{\mathbf{0}\}$, where
\begin{equation*}
  P_{\lspan{Y}|\lspan{X}}(\lspan{0}|\lspan{1}) = P_{Y|X}(\mathbf{0} | \bX).
\end{equation*}

 Since
$P_{\lspan{Y}|\lspan{X}}(\lspan{1}|\lspan{0}) = 0$, the subspace
degradations of $\loc(H,1)$ are Z-channels with the crossover
probability given by $P_{\lspan{Y}|\lspan{X}}(\lspan{0}|\lspan{1})$.
We know that the capacity of Z-channel is a decreasing function of the
crossover probability. So the best subspace degradation is the one
with the smallest $P_{\lspan{Y}|\lspan{X}}(\lspan{0}|\lspan{1})$.
Therefore, the best subspace degradation can be found by evaluating
$P_{Y|X}(\mathbf{0} | \bX)$ for $\bX \in \ffield^{1\times
  M}\setminus \{\mathbf{0}\}$. 
Since different $\bX$ spanning the same row space only need to be
calculated once, we need to consider $|\Gr(1,\ffield^M)| =
\frac{q^M-1}{q-1}$ inputs. 

Since the input/output of a subspace degradation is binary, the maximum
achievable rate of subspace coding for $\loc(H,1)$ is at most $1$ bit
per use, which is much smaller than the lower bound characterized in Example~\ref{ex:j}.
\end{example}

\subsection{Lower Bound on Subspace Coding Capacity}
\label{sec:decomp}

Since it is difficult to find an optimal subspace degradation in
general, we consider in this section the achievable rate of subspace
coding for uniform-given-row-space input distributions to get a lower
bound on the subspace coding capacity. 
We will show (in the next subsection) that the lower
bound to be obtained is exactly the subspace coding capacity for
certain important special cases.

\begin{theorem}\label{the:diq2}
  For a LOC with uniform-given-row-space input distributions,
  \begin{equation} \label{eq:i8gaqif}
    \mutual(\lspan{X};\lspan{Y}) =  \locrate(\rank(X);\rank(Y)) + \mutual(\rank(X);\rank(Y)),
  \end{equation}
  where $\locrate(\rank(X);\rank(Y))$ is defined in \eqref{eq:ratej};
  and hence
  \begin{IEEEeqnarray*}{rCl}
    C_{\subs} 
    & \geq &
    \max_{p_{\lspan{X^\tr}}}[
    \locrate(\rank(X);\rank(Y)) + \mutual(\rank(X);\rank(Y))].
  \end{IEEEeqnarray*}
\end{theorem}
\begin{IEEEproof}
  Fix a uniform-given-row-space input $p_X$.
  For $U \in \Gr(m, \ffield^t)$, let
  \begin{equation*}
  A(m,U) = \{\bX\in \ffield^{t\times m}:\lspan{\bX} = U\}.
  \end{equation*}
  The set $A(m,U)$ has several properties that will be used in the proof.
  For a full-column-rank matrix $\mathbf B$ with $\lspan{\mathbf B}
  = U$, we have
  \begin{equation*}%
  A(m,U) = \{\mathbf{BD}:\mathbf D\in \Fr(\ffield^{\dim(U)\times
    m})\} =  \mathbf{B} \Fr(\ffield^{\dim(U)\times m}).
  \end{equation*}
  Thus, $|A(m,U)| = |\Fr(\ffield^{\dim(U)\times m})| = \cmat{m}{\dim(U)}$.  For
  $\Phi\in \Fr(\ffield^{t\times t})$, $\lspan{\Phi \mathbf B} = \Phi
  U$.  So $A(m,\Phi U) = \Phi \mathbf B \Fr(\ffield^{r\times M}) =
  \Phi A(m,U)$.

   Fix any $V,V', U,U' \in \Pj(\ffield^T)$ satisfying
    $V\leq U$, $V'\leq U'$, $\dim(U)=\dim(U')=r$ and
    $\dim(V)=\dim(V')=s$.  We show that there exists a full rank
    $T\times T$ matrix such that $\Phi V = V'$ and $\Phi U = U'$.
    Find a basis $\{\mathbf{b}_{i}:i=1,\cdots, s\}$ of $V$, extend the
    basis of $V$ to a basis $\{\mathbf{b}_{i}:i=1,\cdots, r\}$ of $U$,
    and further extend the basis of $U$ to a basis
    $\{\mathbf{b}_{i}:i=1,\cdots, T\}$ of $\ffield^T$.  Similarly,
    find a basis $\{\mathbf{b}'_{i}:i=1,\cdots, T\}$ of $\ffield^T$
    such that $\{\mathbf{b}'_{i}:i=1,\cdots, r\}$ is a basis of $U$
    and $\{\mathbf{b}'_{i}:i=1,\cdots, s\}$ is a basis of $V$.
    Consider the linear equation system
  \begin{equation*}
    \Phi \mathbf{b}_i = \mathbf{b}_{i}',\quad i=1,\cdots,T.
  \end{equation*}
  The unique solution of the above system satisfies $\Phi V = V'$ and  $\Phi U = U'$.
  Using the above notations, we have
  \begin{IEEEeqnarray*}{rCl}
    p_{\lspan{X}\lspan{Y}}(U,V)
    & = & {\sum_{\bX\in A(M,U)} p_{X}(\bX)
      \sum_{\bY\in A(N,V)} P_{Y|X}(\bY|\bX)} \\
    & = & {\sum_{\bX\in A(M,U)} p_{X}(\Phi\bX)
      \sum_{\bY\in A(N,V)} P_{Y|X}(\Phi\bY|\Phi\bX)} \IEEEyesnumber \IEEEeqnarraynumspace \label{eq:kdka1}\\
    & = & {\sum_{\bX\in A(M,\Phi U)} p_{X}(\bX)
      \sum_{\bY\in A(N,\Phi V)} P_{Y|X}(\bY|\bX)} \\
    & = & p_{\lspan{X}\lspan{Y}}(\Phi U,\Phi V) \\
    & = & p_{\lspan{X}\lspan{Y}}(U',V'),
  \end{IEEEeqnarray*}
  where in \eqref{eq:kdka1} $p_X(\bX)=p_X(\Phi\bX)$ follows that $p_X$
  is uniform-given-row-space, and $P_{Y|X}(\Phi\bY|\Phi\bX) = P_{Y|X}(\bY|\bX)$
  follows from Lemma~\ref{the:symm}.  Then it can be verified that for
  $V, U  \leq \ffield^T $ with $V\leq U$, $\dim(U)=r$ and
  $\dim(V)=s$,
  \begin{equation}\label{eq:uv}
  p_{\lspan{X}\lspan{Y}}(U,V) = \frac{p_{\rank(X)\rank(Y)}(r,s)}{\gco{T}{r}\gco{r}{s}}.
  \end{equation}
  Similarly, we can show that $p_{\lspan{X}}(U) = p_{\lspan{X}}(U')$
  for $U,U' \leq \ffield^T$ with $\dim(U)=\dim(U')=r$, which implies
  \begin{equation}\label{eq:u}
  p_{\lspan{X}}(U) = \frac{p_{\rank(X)}(r)}{\gco{T}{r}}.
 \end{equation}

Moreover, for $V \leq \ffield^T$ with $\dim(V)=s$,
\begin{IEEEeqnarray*}{rCl}
  p_{\lspan{Y}}(V)
  & = & \sum_{r\geq s} \sum_{U:V\leq U, \dim(U)=r}
   p_{\lspan{X}\lspan{Y}}(U,V)\\
   & = & \sum_{r\geq s}  \frac{p_{\rank(X)\rank(Y)}(r,s)}{\gco{T}{r}\gco{r}{s}} \sum_{U:V\leq U,\dim(U)=r}1 \\
   & = & \sum_{r\geq s} \frac{p_{\rank(X)\rank(Y)}(r,s)}{\gco{T}{r}\gco{r}{s}} \gcos{T}{r}\frac{\cmat{r}{s}}{\cmat{T}{s}} \IEEEyesnumber \label{eq:ksc}\\
   & = & \frac{p_{\rank(Y)}(s)}{\gco{T}{s}}, \IEEEyesnumber \label{eq:v}
\end{IEEEeqnarray*}
where \eqref{eq:ksc} is obtained by Lemma~\ref{lemma:c1}.

Substituting
\eqref{eq:uv}, \eqref{eq:u} and \eqref{eq:v} into
$\mutual(\lspan{X};\lspan{Y})$ completes the proof.
\end{IEEEproof}

\subsubsection{Optimal Uniform-Given-Row-Space Input Distribution for
  Subspace Coding}

Define
\begin{equation*}
  C_{\text{USS}} \triangleq \max_{p_{\lspan{X^\tr}}}[
    \locrate(\rank(X);\rank(Y)) + \mutual(\rank(X);\rank(Y))],
\end{equation*}
which is the maximum achievable rate of subspace coding using
uniform-given-row-space input distribution. Rewrite 
$p_{\lspan{X^\tr}}(U) = p_{\rank(X)}(\dim(U))
P_{\lspan{X^\tr}|\rank(X)}(U|\dim(U))$. By treating $p_{\rank(X)}$ and
$P_{\lspan{X^\tr}|\rank(X)}$ as variables, we can rewrite the above
maximization problem as
\begin{equation}
  \label{eq:opt10}
  \max_{p_{\rank(X)}}\max_{P_{\lspan{X^\tr}|\rank(X)}}
  [\mutual(\rank(X);\rank(Y)) + \locrate(\rank(X);\rank(Y))].
\end{equation}

Both $\mutual(\rank(X);\rank(Y))$ and $\locrate(\rank(X);\rank(Y))$
depend on $P_{\rank(Y)|\rank(X)}$. We have 
\begin{IEEEeqnarray*}{rCl}
  P_{\rank(Y)|\rank(X)}(s|r)
  & = & \sum_{U\in
    \Gr(r,\ffield^M)} P_{\rank(Y)|\lspan{X^\tr}}(s|U)P_{\lspan{X^\tr}|\rank(X)}(U|r),
\end{IEEEeqnarray*}
in which $P_{\rank(Y)|\lspan{X^\tr}}(s| U)$ is a function of $p_H$ and
is not related to $p_{\rank(X)}$ and $P_{\lspan{X^\tr}|\rank(X)}$
(cf. Lemma~\ref{lemma:cond}).
The formulation of
$\locrate(\rank(X);\rank(Y))$ can be rewritten as
\begin{IEEEeqnarray*}{rCl}
\locrate(\rank(X);\rank(Y))
  & = & \IEEEyesnumber
   \sum_r p_{\rank(X)}(r) \sum_{U \in \Gr(r,\ffield^M)}
   P_{\lspan{X^\tr}|\rank(X)} (U|r) R_{H,T}(U), \IEEEeqnarraynumspace \label{eq:cks}
\end{IEEEeqnarray*}
where
\begin{equation}\label{eq:g}
  R(U) = R_{H,T}(U) \triangleq \sum_s P_{\rank(Y)|\lspan{X^\tr}}(s|U) \log \frac{\cmat{T}{s}}{\cmat{\dim( U)}{s}}
\end{equation}
is only related to the distribution of $H$.
Note that $R(U)$ is the achievable rate of subspace
  coding for the uniform-given-row-space input with $p_{\langle X^\top\rangle}(U)=1$.

  For fixed $p_{\rank(X)}$, $\mutual(\rank(X);\rank(Y))$ is a convex
  function of $P_{\lspan{X^\tr}|\rank(X)}$, and
  $\locrate(\rank(X);\rank(Y))$ is a linear function of
  $P_{\lspan{X^\tr}|\rank(X)}$. For fixed
  $P_{\lspan{X^\tr}|\rank(X)}$, $\mutual(\rank(X);\rank(Y))$ is a
  concave function of $p_{\rank(X)}$, and
  $\locrate(\rank(X);\rank(Y))$ is a linear function of
  $p_{\rank(X)}$.  Therefore, $[\mutual(\rank(X);\rank(Y)) +
  \locrate(\rank(X);\rank(Y))]$ is not concave for $p_{\rank(X)}$ and
  $P_{\lspan{X^\tr}|\rank(X)}$.  The following lemma characterizes a
  special optimizer of \eqref{eq:opt10}.

\begin{lemma}\label{the:optssd}
  There exists a deterministic transition matrix
  $P_{\lspan{X^\tr}|\rank(X)}$ achieving $C_{\text{USS}}$.
\end{lemma}
\begin{IEEEproof}
  The proof is similar to the one of Lemma~\ref{the:beta}, and hence omitted.
\end{IEEEproof}

\subsubsection{Optimal Uniform-Given-Row-Space Input Distribution for
  Large $T$ and $q$}

We can further
narrow down the range to search an optimal uniform-given-row-space input
distribution when both $T$ and $q$ are large.
For a random matrix $H$, define
\begin{equation*}
  \rank^*(H) \triangleq \max\{r:\Pr\{\rank(H)=r\}>0\}.
\end{equation*}

\begin{theorem}\label{the:mmm}
  There exists $T_0$ and $R_0$ as functions of $M$ and the rank
  distribution of $H$, such that
  when $T\geq T_0$ and $(T-M)\log q \geq R_0$, $C_{\text{USS}}$ is achieved by the uniform-given-row-space
  input distribution with $\Pr\{\rank(X) \geq \rank^*(H)\} = 1$.
\end{theorem}
\begin{IEEEproof}%
  By Lemma~\ref{the:optssd}, there exists a uniform-given-row-space input
  achieving $C_{\text{USS}}$ such that
  $p_{\lspan{X^\tr}}(U(r))=p_{\rank(X)}(r)$ for all $r\leq
  \min\{M,T\}$, where $\dim(U(r))=r$.  In other words, for $r$ such
  that $p_{\rank(X)}(r)>0$, $P_{\lspan{X^\tr}|\rank(X)}(U(r)|r)=1$.
  We show by contradiction that $\Pr\{\rank(X) \geq \rank^*(H)\} = 1$
  for sufficiently large $T$.

  Consider an input distribution with $\Pr\{\rank(X) < \rank^*(H)\} > 0$. By Theorem~\ref{the:diq2} and
  \eqref{eq:cks},
  \begin{equation*}
    C_{\text{USS}} = \mutual(\rank(X);\rank(Y)) + \sum_r p_{\rank(X)}(r) R(U(r)).
  \end{equation*}
  Define a uniform-given-row-space input distribution $p'_X$ with $p'_{\lspan{X^\tr}}(U(r))=p'_{\rank(X)}(r)=p_{\rank(X)}(r)$ for $\rank^*(H)\leq r <M$ and $p'_{\lspan{X^\tr}}(U(M))=p'_{\rank(X)}(M)=p_{\rank(X)}(M)+\sum_{k<\rank^*(H)}p_{\rank(X)}(k)$.
  We have that
  \begin{IEEEeqnarray*}{rCl}
    \mutual(\lspan{X};\lspan{Y})|_{p_X'}-C_{\text{USS}}
    & \geq &
    \sum_{r=0}^{\rank^*(H)-1} p_{\rank(X)}(r)[R(\ffield^M)-R(U(r))] \\
    & & - \mutual(\rank(X);\rank(Y))|_{p_X} \\
    & > &
    \sum_{r=0}^{\rank^*(H)-1} p_{\rank(X)}(r)\Theta(T,r,H)\log  q \\
    & &  \IEEEyesnumber \label{eq:ckss}
    - \mutual(\rank(X);\rank(Y))|_{p_X},
  \end{IEEEeqnarray*}
  where the last inequality follows from Lemma~\ref{lemma:food2} in  Appendix~\ref{sec:prooffood2} with
  \begin{IEEEeqnarray*}{rCl}
    \Theta(T,r,H) & = & (T-M) \sum_{k: k>r} \Pr\{\rank(H)
    \geq k\} \\
    & &
    -  r(M-r) + \log_q \cmatt{r}{r}.
  \end{IEEEeqnarray*}
  The quantity $\Theta(T,r,H)$ is a lower bound on 
  $(R(\ffield^M)-R(U))/\log q$ with $\dim(U)=r$ and it is positive when $T$ is
  sufficiently large.

  Fix a sufficiently large $T$ such that $\Theta(T,r,H)>0$ for
  $r<\rank^*(H)$. Since $\Pr\{\rank(X) < \rank^*(H)\} > 0$ by
  assumption,   we see that when
  $(T-M)\log q$ is sufficiently large, the RHS of \eqref{eq:ckss} becomes
  positive, a contradiction to $C_{\subs}(H,T)\geq
  \mutual(\lspan{X};\lspan{Y})$ for any input distribution.
\end{IEEEproof}

\subsubsection{Constant-Rank Uniform-Given-Row-Space Input Distributions}
\label{sec:constr}

An input
distribution with $p_{\rank(X)}(r) = 1$ is called a
\emph{constant-rank or rank-$r$ input distribution}.  Note that for a
subspace degradation, using rank-$r$ input is corresponding to using
$r$-dimensional subspace coding.

For a constant-rank uniform-given-row-space input distribution, we always have
$\mutual(\rank(X);\rank(Y))=0$. So, together with \eqref{eq:i8gaqif}, an optimal
contant-rank uniform-given-row-space input distribution for subspace coding can
be found by maximizing $\locrate(\rank(X);\rank(Y))$.
Define
\begin{IEEEeqnarray*}{rCl}
  C_{\text{CUSS}} 
  & \triangleq & 
  \max_{p_X:\text{constant-rank uniform-given-row-space}}
  \locrate(\rank(X);\rank(Y)) \\
  & = & \IEEEyesnumber \label{eq:dds}
  \max_{U\in \Pj(\min\{M,T\},\ffield^M)}R(U),
\end{IEEEeqnarray*}
where \eqref{eq:dds} follows from \eqref{eq:cks}. 

Since $\mutual(\rank(X);\rank(Y)) \leq
\log (\min\{T,M,N\}+1)$, the loss of rate by using constant-rank uniform-given-row-space
input distribution is small when
\begin{equation}
  \label{eq:3}
  C_{\text{CUSS}}  \gg \log (\min\{T,M,N\}+1).
\end{equation}
By Lemma~\ref{lemma:c5}, we know that when $T>M$,
\begin{equation*}
  C_{\text{CUSS}} > (T-M)\E[\rank(H)]\log q.
\end{equation*}
So when $(T-M)\E[\rank(H)]\log q \gg \log (\min\{T,M,N\}+1)$, it is
reasonable to use constant-dimensional subspace coding.

\begin{example}
  Consider $T-1=M=N=64$, $\E[\rank(H)]=32$, and $q=256$. We can
  calculate that $\locrate(\rank(X);\rank(Y)) > 256$, while $\log
  (\min\{T,M,N\}+1) \approx 5$. So the loss of rate by using
  constant-rank uniform-given-row-space input distribution is small.
\end{example}

The following corollary is a direct result of Theorem~\ref{the:mmm} with
the condition that $\rank^*(H)=M$.

\begin{corollary}\label{cor:c1}
  For a transfer matrix $H$ with $\rank^*(H)=M$,
  when both $T$ and $(T-M)\log q$ are sufficiently large, the optimal value of \eqref{eq:opt10} is achieved by the uniform-given-row-space
  input with $p_{\rank(X)}(M) = 1$, and the optimal value is
  $R(\ffield^M) = \sum_s p_{\rank(H)}(s)
  \log \frac{\cmat{T}{s}}{\cmat{M}{s}}$.
\end{corollary}

\subsection{LOCs with a Unique Subspace Degradation}
\label{sec:usd}

Now let us turn to LOCs with a unique subspace
degradation, i.e., $P_{\lspan{Y}|\lspan{X}}$ is invariant with respect to
$P_{X|\lspan{X}}$. For such LOCs, we do not have the issue of
finding an optimal subspace degradation---a subspace $U$ can be
converted to any matrix $\bX$ with $\lspan{\bX} = U$. 
This property makes it easier to apply subspace coding on
LOCs with a unique subspace degradation. As we will further show in
this paper, all LOCs studied in existing literature have a unique
subspace degradation, and some results previous obtained for
special cases are actually shared by
all LOCs with a unique subspace degradation.

By definition, a LOC has a unique
subspace degradation if and only if for any $V$,
\begin{equation}\label{eq:id90sl}
    P_{\lspan{Y}|X}(V|\bX) = P_{\lspan{Y}|X}(V|\bX')\ \text{whenever}\ \lspan{\bX'}=\lspan{\bX}.
\end{equation}
If $H$ is uniform-given-row-space, then the transition matrix of
$\loc(H,T)$ satisfies \eqref{eq:id90sl}, and hence $\loc(H,T)$ has a
unique subspace degradation.  Therefore, the LOCs studied in
\cite{nobrega11, nobrega11a} with uniform-given-rank transfer matrices
have a unique subspace degradation.  Since a row-space-symmetric LOC
has a uniform-given-row-space transfer matrix when $T\geq M$ (see
Lemma~\ref{lemma:new1}), we have the following lemma.
\begin{lemma}\label{lemma:tmrss}
  When $T\geq M$, a row-space-symmetric LOC has a unique subspace degradation.
\end{lemma}

When $T<M$, a row-space-symmetric LOC may not have a unique subspace
degradation.

\begin{example} 
Consider $\loc(H,1)$.
By \eqref{eq:id90sl}, $\loc(H,1)$ has a unique subspace
degradation if and only if for any nonzero $x_1, x_2\in
\ffield^{1\times M}$,
\begin{IEEEeqnarray}{rCcCl}
  P_{\lspan{Y}|X}(\lspan{0}|x_1) & = & P_{\lspan{Y}|X}(\lspan{0}|x_2)
  \label{eq:row1dsi}\\
  P_{\lspan{Y}|X}(\lspan{1}|x_1) & = & P_{\lspan{Y}|X}(\lspan{1}|x_2). \label{eq:row2dsi}
\end{IEEEeqnarray}
However, \eqref{eq:row1dsi} implies \eqref{eq:row2dsi} since
\begin{equation*}
  P_{\lspan{Y}|X}(\lspan{1}|x) = 1 -
  P_{\lspan{Y}|X}(\lspan{0}|x).
\end{equation*}
The equalities in \eqref{eq:row1dsi} give linear constraints on
the distribution of $H$, from which we can find the set of $H$ such
that $\loc(H,1)$ has a unique subspace degradation.
\end{example}

More examples of LOCs with a unique subspace degradation will be
provided in Section~\ref{sec:34}.

\begin{lemma}\label{lemma:uniquesd}
A LOC has a unique subspace degradation if and only if
\begin{equation}\label{eq:id90sld}
  P_{\lspan{Y}|X}(V|\bX) = P_{\lspan{Y}|X}(V'|\bX')
\end{equation}
whenever $\dim(V)=\dim(V')$, $\rank(\bX)=\rank(\bX')$, $V\leq
\lspan{\bX}$ and $V'\leq\lspan{\bX'}$.
\end{lemma}
\begin{IEEEproof}
  The sufficient condition holds since \eqref{eq:id90sld}
  implies \eqref{eq:id90sl}. We prove the necessary condition as follows.
  Fix a full column-rank matrix
  $\mathbf{B}_0$ such that $\lspan{\mathbf{B}_0} = V$. Since
  $V\leq \lspan{\bX}$, we can find full rank matrix $\mathbf{B}_1$ and $\mathbf{D}$ such that
  $[\mathbf{B}_0 \mathbf{B}_1]\mathbf{D} = \bX$. Therefore,
  \begin{IEEEeqnarray*}{rCl}
    P_{\lspan{Y}|{X}}(V|\bX)
    & = &
    P_{\lspan{Y}|X}\left(V|[\mathbf{B}_0 \mathbf{B}_1]
    \mathbf{D}\right) \\
    & = &
    \sum_{\bY:\lspan{\bY}=V} P_{Y|\lspan{X}}\left(\bY |[\mathbf{B}_0 \mathbf{B}_1]
    \mathbf{D}\right)\\
    & = &
    \sum_{\mathbf{E}\in \Fr(\ffield^{\dim(V)\times N})}
    P_{Y|X}\left([\mathbf{B}_0 \mathbf{B}_1] \begin{bmatrix} \mathbf{E} \\
      \mathbf{0} \end{bmatrix} \Bigg|[\mathbf{B}_0 \mathbf{B}_1]
    \mathbf{D}\right) \\
    & = & \IEEEyesnumber \label{eq:yam3}
     \sum_{\mathbf{E}\in \Fr(\ffield^{\dim(V)\times N})}
     \Pr\left\{   \mathbf{D} H = \begin{bmatrix} \mathbf{E} \\
      \mathbf{0} \end{bmatrix} \right\},
  \end{IEEEeqnarray*}
  where   \eqref{eq:yam3} follows from Lemma~\ref{the:symm}.
  If \eqref{eq:id90sl} holds, then \eqref{eq:yam3} holds for any full row-rank
  $\rank(\bX)\times M$ matrix $\mathbf{D}$, and hence
  \eqref{eq:id90sld} holds.
\end{IEEEproof}

Recall the definition of uniform-given-dimension distributions over
$\Pj(\ffield^T)$ in Definition~\ref{def:ugd}. 

\begin{theorem}\label{the:uniquesd} 
  For a LOC with a unique subspace degradation, the capacity of the
  subspace degradation can be achieved by a uniform-given-dimension distribution, and 
\begin{IEEEeqnarray*}{rCl}
    C_{\mathrm{SS}} 
    & = & \IEEEyesnumber \label{eq:opt133}
    \max_{p_{\rank(X)}} [\locrate(\rank(X);\rank(Y)) + \mutual(\rank(X);\rank(Y))].
\end{IEEEeqnarray*}
\end{theorem}
\begin{IEEEproof}
For a LOC with a unique subspace degradation,
$P_{\lspan{Y}|\lspan{X}}$ is well defined without specifying
$p_{X|\lspan{X}}$. So considering $p_{\lspan{X}}$ is sufficient for $\mutual(\lspan{X};\lspan{Y})$.
We first show that there exists a uniform-given-dimension input
distribution that maximizes $I(\lspan{X};\lspan{Y})$.

  Fix a LOC with a unique subspace degradation.
  Let $p$ be a distribution over $\Pj(\ffield^T)$ achieving the
  capacity of the subspace degradation, i.e., $p$ achieves $C_\subs$.
  For $\Phi\in \Fr(\ffield^{T\times T})$, define $p^{\Phi}$ as
  $p^{\Phi}(U) = p(\Phi U)$, where $\Phi U$ is defined
  in \eqref{eq:msm}.
  First $p^{\Phi}$ is a PMF because $0\leq p^{\Phi}(U) =p(\Phi U) \leq 1$
  and $\sum_{U\in \Pj(\ffield^{T})} p^{\Phi}(U) = 1$.

  We show that $p^{\Phi}$ also achieves the capacity. For the
  simplicity of the notations, we write $p' = p^{\Phi}$.
  Let $p_{\lspan{Y}}$ and $p_{\lspan{Y}}'$ be the PMF of $\lspan{Y}$ when the input distributions
  are $p$ and $p'$, respectively. We have
  \begin{IEEEeqnarray*}{rCl}
  p'_{\lspan{Y}}(V)
	& = &
	\sum_{U \in \Pj(\ffield^{T}):V\leq U} p'(U) P_{\lspan{Y}|\lspan{X}}(V|U) \\
	& = & \IEEEyesnumber \label{eq:osaa-1}
	\sum_{U \in \Pj(\ffield^{T}):V\leq U} p(\Phi U)
        P_{\lspan{Y}|\lspan{X}}(\Phi V|\Phi U) \\
	& = & %
        \sum_{U' \in \Pj(\ffield^{T}):\Phi V\leq U'} p(U')
        P_{\lspan{Y}|\lspan{X}}(\Phi V|U') \\
	& = &
       p_{\lspan{Y}}(\Phi V),
  \end{IEEEeqnarray*}
  where \eqref{eq:osaa-1} follows from $p'(U)=p(\Phi U)$ and
  Lemma~\ref{lemma:uniquesd}. Therefore,
  \begin{IEEEeqnarray*}{rCl}
    \mutual(\lspan{X};\lspan{Y})|_{p'}
	& = &
	\sum_{U \in \Pj(\ffield^{T})}p'(U) \sum_{V \in
          \Pj(\ffield^{T}): V\leq U} P(V|U)
  \log  \frac{P(V|U)}{p'_{\lspan{Y}}(V)}  \\
	& = & %
	\sum_{U \in \Pj(\ffield^{T})}p(\Phi U) \sum_{V \in
          \Pj(\ffield^{T}): V\leq U} P(\Phi V |\Phi U) \\
        & & \times \log
        \frac{P(\Phi V|\Phi U)}{p(\Phi V)} \\
	& = &
	\sum_{U' \in \Pj(\ffield^{T})}p(U') \sum_{V' \in
          \Pj(\ffield^{T}): V' \leq U'} P(V'|U') \log
        \frac{P(V'|U')}{p(V')} \\
	& = &
       \mutual(\lspan{X};\lspan{Y})|_{p},
 \end{IEEEeqnarray*}
 which implies that $p'$ also achieves the subspace coding capacity.

 Define $p^*$ on $\Pj(\ffield^T)$ as
 \begin{equation*}
   p^*(U) = \frac{1}{|\Fr(\ffield^{T\times T})|}\sum_{\Phi\in \Fr(\ffield^{T\times T})}
 p^{\Phi}(U).
 \end{equation*}
 Note that $p^*$ is uniform-given-dimension.
 Since mutual information is a concave function of the input
 distribution \cite{gallager},
 \begin{equation*}
    \mutual(\lspan{X};\lspan{Y})|_{p^*} \geq \frac{1}{|\Fr(\ffield^{T\times T})|}\sum_{\Phi\in \Fr(\ffield^{T\times T})}
    \mutual(\lspan{X};\lspan{Y})|_{p^{\Phi}}.
 \end{equation*}
Thus, $p^*$ is also an optimal input distribution for the subspace channel.

Note that for a uniform-given-dimension LOC,
\begin{equation*}
  p_{\lspan{X}}(V)= \frac{p_{\rank(X)}(\dim(V))}{\gcos{T}{\dim(V)}}.
\end{equation*}
So $C_{\subs}$ can be found by only optimizing over the
input rank distribution $p_{\rank(X)}$.

If $X$ is a uniform-given-row-space distribution, then $\lspan{X}$ is
uniform-given-dimension.  For any uniform-given-dimension distribution
$p$ on $\Pj(\ffield^T)$ we can find a uniform-given-row-space
distribution $p'$ on $\ffield^{T\times M}$ such that
$p(U)=\sum_{\bX:\lspan{\bX}=U}p'(\bX)$.  Hence, by
Theorem~\ref{the:diq2} and the fact that a uniform-given-row-space input distribution $p_X$ can be
determined by $p_{\lspan{X^\tr}}$, we get
\begin{IEEEeqnarray*}{rCl}
  C_{\subs}
     & = & 
     \max_{p_{\lspan{X^\tr}}}[\locrate(\rank(X);\rank(Y)) + \mutual(\rank(X);\rank(Y))].
\end{IEEEeqnarray*}

Fix $U, U' \leq \ffield^M$ with $\dim(U)=\dim(U')$.
Find $\bX_U$ and $\bX_{U'}$ with $\lspan{\bX_U}=\lspan{\bX_{U'}}$,
$\lspan{\bX_U^\tr}=U$ and $\lspan{\bX_{U'}^\tr}=U'$.
By Lemma~\ref{lemma:cond},
$P_{\rank(Y)|\lspan{X^\tr}}(s|U) = P_{\rank(Y)|X}(s|\bX_U)$ and
$P_{\rank(Y)|\lspan{X^\tr}}(s|U') = P_{\rank(Y)|X}(s|\bX_{U'})$.
Further by Lemma~\ref{lemma:uniquesd}, 
\begin{IEEEeqnarray*}{rCl}
  P_{\rank(Y)|X}(s|\bX_U) & = &
\sum_{V \in \Gr(s, \lspan{\bX_U})} P_{\lspan{V}|X}(V|\bX_U) \\
 & = &
\sum_{V \in \Gr(s, \lspan{\bX_{U'}})} P_{\lspan{V}|X}(V|\bX_{U'}) \\
 & = & P_{\rank(Y)|X}(s|\bX_{U'}).
\end{IEEEeqnarray*}
Therefore,   $P_{\rank(Y)|\lspan{X^\tr}}(s|U)=P_{\rank(Y)|\lspan{X^\tr}}(s|U')$.
Hence, $P_{\rank(Y)|\rank(X)}$ only depends on the
distribution of $H$, and hence $p_{\rank(X)\rank(Y)}$ depends on
${\lspan{X^\tr}}$ only through  ${\rank(X)}$.
The proof is completed.
\end{IEEEproof}

The above theorem implies that input distributions $p_X$ with $p_{\lspan{X}}$
being uniform-given-dimension achieve the subspace coding capacity for
LOCs with a unique subspace degradation. 
Since only the input rank affects the subspace coding capacity, it has
no penalty if we only consider uniform-given-rank input distributions
for subspace coding.

Now, consider the computation of $C_\subs$ for LOCs with a unique
subspace degradation, i.e., solving the maximization in
\eqref{eq:opt133}. The problem is simpler than the one of
computing $C_{\text{USS}}$ (see \eqref{eq:opt10}) since we do not
need to optimize $P_{\lspan{X^\tr}|\rank(x)}$. The 
proof of Theorem~\ref{the:uniquesd} implies
\begin{equation}
  \label{eq:di8s}
  P_{\rank(Y)|\rank(x)}(s|r) = P_{\rank(Y)|\lspan{X^\tr}}(s|U) = 
P_{\rank(Y)|X}(s|\bX)
\end{equation}
for any $U\in \Gr(r,\ffield^M)$ and any $\bX$ with $\rank(\bX)=r$.
The optimization in \eqref{eq:opt133} is convex
and has $\min\{M,T\}$ variables.

Similar to $R_{H,T}(U)$ (defined in \eqref{eq:g}), by abuse of
notations, we define for LOCs with a unique subspace degradation
\begin{equation*}
  R(r) = R_{H,T}(r) = \sum_s P_{\rank(Y)|\rank(X)}(s|r) \log \frac{\cmat{T}{s}}{\cmat{r}{s}}.
\end{equation*}
Actually, $R_{H,T}(\dim(U)) = R_{H,T}(U)$ and hence we can rewrite
\begin{equation*}%
  \locrate(\rank(X);\rank(Y)) = \sum_r p_{\rank(X)}(r) R(r).
\end{equation*}
When applying on LOC with a unique subspace degradations, the same
result of Theorem~\ref{the:mmm} still holds (with $C_{\subs}$ in place of
$C_{\text{USS}}$) and the proof can be simplified by using $R(r)$ in
stead of $R(U)$.
Similar to the discussion around \eqref{eq:dds}, the maximum
achievable rate of constant-rank input distributions is
\begin{equation*}
  \max_{p_{\rank(X)}} \locrate(\rank(X);\rank(Y)) = \max_{r}R(r).
\end{equation*}

\begin{example}
  Let's apply the above general discussion on LOCs with
  uniform-given-rank transfer matrices. Assume that $p_{\rank(H)}$ 
  is known. To compute
  $P_{\rank(Y)|\rank(X)}(s|r)$, we choose the input matrix
  \begin{equation*}
    \bX^{(r)} =
    \begin{bmatrix}
      I_r & \mathbf{0} \\ \mathbf{0} & \mathbf{0}
    \end{bmatrix}.
  \end{equation*}
  For transfer matrix $H = \begin{bmatrix} H_1 \\
    H_2 \end{bmatrix}$ where $H_1$ has $r$ rows and $H_2$ has
  $M-r$ rows, the output matrix is $\begin{bmatrix} H_1 \\
    \mathbf{0} \end{bmatrix}$. So
  \begin{IEEEeqnarray*}{rCl}
    P_{\rank(Y)|\rank(X)}(s|r)
    & = & 
    P_{\rank(Y)|X}(s|\bX^{(r)}) \\
    & = & 
    \Pr\{ \rank(H_1) = s \} \\
    & = & 
    \sum_{k=s}^{\min\{M,N\}} \Pr\{ \rank(H_1) = s | \rank(H)=k \}
    p_{\rank(H)}(k).    
  \end{IEEEeqnarray*}
  Since the transfer matrix is uniform-given-rank, we have
  \begin{IEEEeqnarray*}{rCl}
    \Pr\{ \rank(H_1) = s | \rank(H)=k \}
    & = & \frac{|\{\bH\in
      \ffield^{M\times N}: \rank(\bH_1)=s, \rank(\bH)=k\}|}{|\{\bH\in
      \ffield^{M\times N}: \rank(\bH)=k\}|},
  \end{IEEEeqnarray*}
  where the RHS can be counted using the techniques introduced in
  Preliminaries. After computing $P_{\rank(Y)|\rank(X)}(s|r)$, $R(r)$
  can be computed accordingly. Then the subspace coding capacity, as well as
  an optimal input rank distribution, can be obtained by solving
  \begin{IEEEeqnarray*}{rCl}
    & \max_{p(r)} \quad & \sum_{r} p(r) R(r) + \sum_{r} p(r) \sum_s
    P_{\rank(Y)|\rank(X)}(s|r) \log 
    \frac{P_{\rank(Y)|\rank(X)}(s|r)}{\sum_{r'} p(r')
      P_{\rank(Y)|\rank(X)}(s|r')} \\
    & \text{s.t.} & p(r)\geq 0, \ \sum_r p(r) = 1.
  \end{IEEEeqnarray*}
  We would not go into the details about solving the above
  optimization problem. Readers are referred to \cite{nobrega11,
    nobrega11a} to find more results about LOCs with
  uniform-given-rank transfer matrices.
\end{example}

\section{Shannon Capacity vs Subspace Coding Capacity}
\label{sec:compare}

In this section, we discuss some necessary conditions and sufficient
conditions for a LOC such that $C=C_{\subs}$ as applications of the
results obtained in the previous sections.  A new class of LOCs such
that $C=C_{\subs}$ is explicitly characterized.

\subsection{Unique Subspace Degradation}

 Theorem~\ref{the:89s} says
  \begin{equation*}
    C \geq C_L \triangleq \max_{p_{\lspan{X^\tr}}} \left[J(\rank(X);\rank(Y)) + I(\lspan{X^\tr};\lspan{Y^\tr})\right],
  \end{equation*}
  and for a LOC with a unique subspace degradation Theorem~\ref{the:uniquesd} shows
  \begin{equation*}
    C_\subs = \max_{p_{\lspan{X^\tr}}}[
    \locrate(\rank(X);\rank(Y)) + \mutual(\rank(X);\rank(Y))].
  \end{equation*}
The above bounds imply a necessary condition such that $C=C_\subs$.

\begin{theorem}\label{the:capacitysd}
  Consider a LOC with a unique subspace degradation. If $C=C_\subs$, then
  for certain $p_{\lspan{X^\tr}}$ that achieves $C_L$,
  $\lspan{X^\tr}\rightarrow \rank(X) \rightarrow \rank(Y)\rightarrow \lspan{Y^\tr}$ is a
  Markov chain.  In other words, subspace coding is not capacity
  achieving if the LOC does not satisfy the Markov condition $\lspan{X^\tr}\rightarrow \rank(X)\rightarrow \rank(Y)\rightarrow \lspan{Y^\tr}$ for
  any $p_{\lspan{X^\tr}}$ achieving $C_L$.
\end{theorem}
\begin{IEEEproof}
  Fix a LOC with a unique subspace degradation and $C=C_\subs$.
  If there is no $p_{\lspan{X^\tr}}$ achieving $C_L$ and $C_\subs$
  simultaneously, $C > C_{\subs}$.
  Consider a distribution $p^*_{\lspan{X^\tr}}$ of $\lspan{X^\tr}$
  that achieves $C_L$ and $C_{\subs}$ simultaneously, for which we have
  $\mutual(\lspan{X^{\tr}}; \lspan{Y^{\tr}}) =
  \mutual(\rank(X);\rank(Y))$,
  which implies $\mutual(\lspan{X^\tr};\lspan{Y^\tr}|\rank(Y))=0$ and
  $\mutual(\lspan{X^\tr};\rank(Y)|\rank(X))=0$. So both
  $\lspan{X^\tr}\rightarrow \rank(X)\rightarrow \rank(Y)$ and $(\lspan{X^\tr},\rank(X))\rightarrow
  \rank(Y)\rightarrow \lspan{Y^\tr}$ form Markov chains. Hence
  $\lspan{X^\tr}\rightarrow \rank(X)\rightarrow \rank(Y)\rightarrow
  \lspan{Y^\tr}$ is a Markov chain.
\end{IEEEproof}

We know that for a distribution $p_{\lspan{X^\tr}}$, $\lspan{X^\tr}\rightarrow \rank(X) \rightarrow
  \rank(Y)\rightarrow \lspan{Y^\tr}$ is a Markov chain if and only if
  \begin{IEEEeqnarray*}{rCl}
   p_{\rank(X)}(r)p_{\rank(Y)}(s)p_{\lspan{X^\tr}\rank(X)\lspan{Y^\tr}\rank(Y)}(U,r,V,s)
    & = &
    p_{\lspan{X^\tr}\rank(X)}(U,r)
    p_{\rank(X)\rank(Y)}(r,s)
    p_{\lspan{Y^\tr}\rank(Y)}(V,s), \\
    &  & \quad \forall r,s,U,V,
  \end{IEEEeqnarray*}
which is equivalent to
\begin{IEEEeqnarray*}{rCl}
  p_{\rank(X)}(r)p_{\rank(Y)}(s)p_{\lspan{X^\tr}\lspan{Y^\tr}}(U,V)
  & = & 
  p_{\lspan{X^\tr}}(U)
    p_{\rank(X)\rank(Y)}(r,s)
    p_{\lspan{Y^\tr}}(V), \\
  & & \quad \forall r \geq s,\dim(U)=r,\dim(V)=s. \IEEEyesnumber   \label{eq:1ll}
\end{IEEEeqnarray*}
For $U$ such that $p_{\lspan{X^\tr}}(U)>0$, the equality in \eqref{eq:1ll}
becomes 
\begin{IEEEeqnarray*}{rCl}
  {p_{\rank(Y)}(s)P_{\lspan{Y^\tr}|\lspan{X^\tr}}(V|U)}
  & = & 
    p_{\rank(Y)|\rank(X)}(s|r)
    p_{\lspan{Y^\tr}}(V),
\end{IEEEeqnarray*}
where $r = \dim(U) \geq \dim(V)=s$. Thus,  for each $V$, among all $U\in
\Gr(r,\ffield^M)$ with $p_{\lspan{X^\tr}}(U)>0$,
$p_{\lspan{Y^\tr}|\lspan{X^\tr}}(V|U)$ are the same.
Therefore, we can have the following lemma. 

\begin{lemma}\label{cor:dd3}
If $p^*_{\lspan{X^\tr}}$ achieves $C_L$ and satisfies the Markov chain
$\lspan{X^\tr}\rightarrow \rank(X) \rightarrow \rank(Y)\rightarrow
\lspan{Y^\tr}$, then
there exists $p'_{\lspan{X^\tr}}$ achieving $C_L$ such that
  \begin{enumerate}
  \item for each $r$ there exists at most one $U_r \in
    \Gr(r,\ffield^M)$ such that $p'_{\lspan{X^\tr}}(U_r)>0$; and
  \item the Markov chain $\lspan{X^\tr}\rightarrow \rank(Y)\rightarrow
\lspan{Y^\tr}$ holds.
  \end{enumerate}
\end{lemma}
\begin{IEEEproof}
Let
\begin{equation*}
  f(p_{\lspan{X^\tr}}) = J(\rank(X);\rank(Y))+
I(\lspan{X^\tr};\lspan{Y^\tr}).
\end{equation*}
Then $C_L =
\max_{p_{\lspan{X^\tr}}} f(p_{\lspan{X^\tr}})$. We have
\begin{IEEEeqnarray*}{rCl} 
  \frac{\partial f(p_{\lspan{X^\tr}})}{\partial p_{\lspan{X^\tr}}(U)}
  & = & \sum_{s = 0}^{\dim(U)} \sum_{V\in \Gr(s,\ffield^N)}
  P(V|U) \log \frac{P(V|U)}{p_{\lspan{Y^\tr}}(V)} 
   + R(U)- \log e. \IEEEyesnumber \label{eq:c9s}
\end{IEEEeqnarray*}
Let $p_{\lspan{X^\tr}}'$ be a
distribution
on $\Pj(\ffield^M)$ such that for each $r$ with $p^*_{\rank(X)}(r)>0$,
there exists $U_r \in \Gr(r, \ffield^M)$ such that $p_{\lspan{X^\tr}}'(U_r) = \sum_{U\in
  \Gr(r, \ffield^M)} p^*_{\lspan{X^\tr}}(U)$ and $p^*_{\lspan{X^\tr}}(U_r)
>0$.

Let $p_{\lspan{Y^\tr}}^*$ and $p_{\lspan{Y^\tr}}'$ be the distribution of $\lspan{Y^\tr}$ with
respect to $p_{\lspan{X^\tr}}^*$ and $p_{\lspan{X^\tr}}'$, respectively. 
We have
\begin{IEEEeqnarray*}{rCl}
  p^*_{\lspan{Y^\tr}}(V)
  & = & 
  \sum_{r} \sum_{U\in \Gr(r,\ffield^M): p_{\lspan{X^\tr}}(U)>0}
  P_{\lspan{Y^\tr}|\lspan{X^\tr}}(V|U) p_{\lspan{X^\tr}}(U) \\
  & = & 
  \sum_{r} \sum_{U\in \Gr(r,\ffield^M): p_{\lspan{X^\tr}}(U)>0}
  P_{\lspan{Y^\tr}|\lspan{X^\tr}}(V|U_r) p_{\lspan{X^\tr}}(U) \\
  & = & 
  \sum_{r} P_{\lspan{Y^\tr}|\lspan{X^\tr}}(V|U_r)
  p'_{\lspan{X^\tr}}(U_r) \\
  & = & 
  p_{\lspan{Y^\tr}}'(V),
\end{IEEEeqnarray*}
where the second equality follows from the discussion after \eqref{eq:1ll}.
By checking the KKT condition \cite{KKT}, $p_{\lspan{X^\tr}}'$
achieves $C_L$ since $\frac{\partial f(p_{\lspan{X^\tr}})}{\partial
  p_{\lspan{X^\tr}}(U)}\Big |_{p_{\lspan{X^\tr}}'} = \frac{\partial f(p_{\lspan{X^\tr}})}{\partial p_{\lspan{X^\tr}}(U)}\Big |_{p^*_{\lspan{X^\tr}}} $.

For $r$ with $p_{\lspan{X^\tr}}(r)>0$, we further have
\begin{IEEEeqnarray*}{rcl}
  p_{\rank(Y)|\rank(X)}(s|r)
  & = & 
  \sum_{U\in\Gr(r,\ffield^M):p_{\lspan{X^\tr}}(U)>0}
  P_{\rank(Y)|\lspan{X^\tr}}(s|U)p_{\lspan{X^\tr}|\rank(X)}(U|r) \\
  & = &
  P_{\rank(Y)|\lspan{X^\tr}}(s|U_r) \\
  & = & 
  p'_{\rank(Y)|\rank(X)}(s|r).
\end{IEEEeqnarray*}
Therefore, $p_{\lspan{X^\tr}}'$ also satisfies the Markov condition in
\eqref{eq:1ll}.
Note that since in this case, the distributions of $\rank(X)$ and
$\lspan{X^\tr}$ are the same, we equivalently have the Markov
condition 
$\lspan{X^\tr}\rightarrow \rank(Y)\rightarrow
\lspan{Y^\tr}$.
\end{IEEEproof}

Using Lemma~\ref{cor:dd3}, the sufficient condition in
Theorem~\ref{the:capacitysd} can be refined, and 
more explicit necessary conditions can be obtained for special cases.

\begin{example}\label{ex:ndd}
  Suppose that $C=C_{\subs}$ for certain $\loc(H,1)$ with a unique subspace
  degradation. Fix $p'_{\lspan{X^\tr}}$ such that 1)
  $p'_{\lspan{X^\tr}}$ achieves $C_L$ and 2)
  $p'_{\lspan{X^\tr}}(U_1) = 1 - p_{\rank(X)}(0)$ for $U_1 \in
  \Gr(1,\ffield^M)$. The existence of such $p_{\lspan{X^\tr}}$ is
  guaranteed by Theorem~\ref{the:capacitysd} and
  Lemma~\ref{cor:dd3}. Using
  \eqref{eq:c9s}, we have for $U\in \Gr(1,\ffield^M)$
\begin{IEEEeqnarray*}{rCl}
  \frac{\partial
      f(p_{\lspan{X^\tr}})}{\partial p_{\lspan{X^\tr}}(U)}\Bigg |_{p =
      p'}
  & = & 
  \sum_{s = 0}^{1} \sum_{V\in \Gr(s,\ffield^N)}
  P(V|U) \log \frac{P(V|U)}{p_{\lspan{Y^\tr}}(V)} - \log e \\
  & = & 
  P(\lspan{\mathbf{0}}|U) \log
  \frac{P(\lspan{\mathbf{0}}|U)}{p_{\lspan{Y^\tr}}(\lspan{\mathbf{0}})}
   + \sum_{V\in \Gr(1,\ffield^N)}
  P(V|U) \log \frac{P(V|U)}{P(V|U_1)p_{\lspan{X^\tr}}(U_1)} - \log e
  \\
  & = & 
  - \log p'_{\lspan{X^\tr}}(U_1) - \log e + D_{\text{KL}}(P(\cdot|U)||P(\cdot|U_1)) + P(\lspan{\mathbf{0}}|U) \log
  \frac{P(\lspan{\mathbf{0}}|U_1) p'_{\lspan{X^\tr}}(U_1)}{p_{\lspan{Y^\tr}}(\lspan{\mathbf{0}})},
\end{IEEEeqnarray*}
where $D_{\text{KL}}$ is the Kullback-Leibler divergence (cf. \cite{yeung08}).
By \eqref{eq:row1dsi}, $P_{\lspan{Y^\tr}|\lspan{X^\tr}}(\lspan{\mathbf{0}}|U) =
P_{\lspan{Y^\tr}|\lspan{X^\tr}}(\lspan{\mathbf{0}}|U_1)$ for all $U\in
\Gr(1,\ffield^M)$. Since $p_{\lspan{X^\tr}}$ achieves $C_L$, by the
KKT condition, we have for all $U \neq U_1$
\begin{equation*}
  \frac{\partial f(p_{\lspan{X^\tr}})}{\partial p_{\lspan{X^\tr}}(U)} 
  \Bigg |_{p = p'}
  \leq \frac{\partial f(p_{\lspan{X^\tr}})}{\partial
    p_{\lspan{X^\tr}}(U_1)} \Bigg |_{p = p'}
\end{equation*}
which implies $D_{\text{KL}}(P(\cdot|U)||P(\cdot|U_1)) = 0$.
Therefore, for $\loc(H,1)$ with a unique subspace degradation, a necessary condition such that subspace
coding is capacity achieving is that for each $V \in
\Gr(1,\ffield^N)$,
$P_{\lspan{Y^\tr}|\lspan{X^\tr}}(V|U)=P_{\lspan{Y^\tr}|\lspan{X^\tr}}(V|U')$
for all $U,U'\in \Gr(1,\ffield^M)$.
\end{example}

We can get a stronger result if the LOC with a unique subspace degradation is also row-space symmetric. Note that when $T\geq M$, a row-space-symmetric LOC has a unique
  subspace degradation (cf. Lemma~\ref{lemma:tmrss}).

\begin{corollary}\label{cor:3}
  For a row-space-symmetric LOC which has a unique subspace
  degradation, $C=C_\subs$ if and only if for certain $p_{\lspan{X^\tr}}$ that achieves $C$,
  $\lspan{X^\tr}\rightarrow \rank(X) \rightarrow \rank(Y)\rightarrow \lspan{Y^\tr}$ is a
  Markov chain.
\end{corollary}
\begin{IEEEproof}
  For a row-space-symmetric LOC, $C=C_L$. So the necessary condition
  follows from Theorem~\ref{the:capacitysd}. On the other hand, assume $\lspan{X^\tr}\rightarrow \rank(X) \rightarrow \rank(Y)\rightarrow \lspan{Y^\tr}$ is a
  Markov chain for certain $p_{\lspan{X^\tr}}$ that achieves
  $C$. So $\mutual(\lspan{X^{\tr}}; \lspan{Y^{\tr}}) =
  \mutual(\rank(X);\rank(Y))$ for this $p_{\lspan{X^\tr}}$.
  Therefore $C = J(\rank(X);\rank(Y)) + I(\lspan{X^\tr};\lspan{Y^\tr})
  = J(\rank(X);\rank(Y)) + \mutual(\rank(X);\rank(Y)) \leq C_\subs$,
  which implies $C=C_\subs$.
\end{IEEEproof}

\begin{example}\label{example:uuu}
  Following Example~\ref{ex:rss}, we discuss $\loc_{2}(H,1)$ with
  a unique subspace degradation. We know $H$ satisfies
  \eqref{eq:row1dsi}. Since
  $\loc_{2}(H,1)$ is row-space-symmetric, we can apply the
   necessary and sufficient for $C=C_\subs$ given in
   Corollary~\ref{cor:3}. Similar to the discussion in Example~\ref{ex:ndd}, we have that for $\loc_{2}(H,1)$ with
  a unique subspace degradation,
  $C=C_\subs$ if and only if for each $y \in \ffield^{1\times N}$
  \begin{equation}\label{eq:cioss}
    P_{Y|X}(y|x_1)=P_{Y|X}(y|x_2), \ \forall x_1,x_2\in
    \ffield^{1\times M}.
  \end{equation}
  We will connect the above condition to another class of LOCs to be
  discussed.
\end{example}

\subsection{Row-Space-Symmetric LOCs ($T<M$)}
\label{sec:oss}
When
$T\geq M$, a row-space-symmetric LOC has a unique subspace degradation
(cf. Lemma~\ref{lemma:tmrss}). Hence, we can apply the results in the
last subsection.  But when $T<M$, a row-space-symmetric LOC may not
have a unique subspace degradation.  The following lemma gives an
upper bound on the subspace coding capacity of row-space-symmetric
LOCs.

\begin{lemma}\label{lemma:rsuloc}
  For a row-space-symmetric LOC,
  \begin{equation*}
    C_{\mathrm{SS}} \leq \max_{p_{\lspan{X^\tr}}}\left[\locrate(\rank(X);\rank(Y)) + \mutual(\lspan{X^\tr};\rank(Y))\right].
  \end{equation*}
\end{lemma}
\begin{IEEEproof}
  See Appendix~\ref{sec:pll}.
\end{IEEEproof}

\begin{theorem}\label{the:non2}
  Consider a row-space-symmetric LOC.
  \begin{enumerate}
  \item (Necessary condition) If $C=C_\subs$, then for certain  $p_{\lspan{X^\tr}}$ that achieves $C$,
  $\lspan{X^\tr}\rightarrow \rank(Y)\rightarrow \lspan{Y^\tr}$ is a
  Markov chain. In other words, subspace coding is not capacity
  achieving if the LOC does not satisfy the Markov condition $\lspan{X^\tr}\rightarrow \rank(Y)\rightarrow \lspan{Y^\tr}$ for
  all $p_{\lspan{X^\tr}}$ achieving $C$.
  \item (Sufficient condition) If for certain  $p_{\lspan{X^\tr}}$ that achieves $C$,
  $\lspan{X^\tr}\rightarrow \rank(X)\rightarrow \rank(Y)\rightarrow \lspan{Y^\tr}$ is a
  Markov chain, then $C=C_\subs$.
  \end{enumerate}
\end{theorem}
\begin{IEEEproof}
We first prove the necessary condition. Fix a row-space-symmetric LOC such that $C=C_\subs$.
By Lemma~\ref{lemma:rsuloc},
 $$C_{\subs} \leq R^U \triangleq
\max_{p_{\lspan{X^\tr}}}\left[\locrate(\rank(X);\rank(Y)) +
  \mutual(\lspan{X^\tr};\rank(Y))\right].$$
 On the other hand,
by Theorem~\ref{the:89s},
  \begin{equation*}
    C = \max_{p_{\lspan{X^\tr}}} \left[J(\rank(X);\rank(Y)) + I(\lspan{X^\tr};\lspan{Y^\tr})\right].
  \end{equation*}
Since $I(\lspan{X^\tr};\lspan{Y^\tr}) \geq
\mutual(\lspan{X^\tr};\rank(Y))$ for any $p_{\lspan{X^\tr}}$, if there
exists no $p_{\lspan{X^\tr}}$ achieving $C$ and $R^U$ simultaneously,
$C > R^U \geq C_{\subs}$, a contradiction to $C=C_\subs$.
Fix $p_{\lspan{X^\tr}}$ that achieves $C$ and $R^U$ simultaneously.
We have $\mutual(\lspan{X^{\tr}}; \lspan{Y^{\tr}}) =
  \mutual(\lspan{X^\tr};\rank(Y))$, which implies
  $I(\lspan{X^\tr};\lspan{Y^\tr}|\rank(Y))=0$, i.e., $\lspan{X^\tr}\rightarrow \rank(Y)\rightarrow \lspan{Y^\tr}$ is a
  Markov chain.

Now we show the sufficient condition. Fix a $p_{\lspan{X^\tr}}$ that achieves
  $C$ and for which $\lspan{X^\tr}\rightarrow \rank(X) \rightarrow \rank(Y)\rightarrow \lspan{Y^\tr}$ is a
  Markov chain. So $\mutual(\lspan{X^{\tr}}; \lspan{Y^{\tr}}) =
  \mutual(\rank(X);\rank(Y))$ for this $p_{\lspan{X^\tr}}$.
  Thus
  \begin{IEEEeqnarray*}{rCl}
    C & = & J(\rank(X);\rank(Y)) + I(\lspan{X^\tr};\lspan{Y^\tr}) \\
    & = & J(\rank(X);\rank(Y)) + \mutual(\rank(X);\rank(Y)) \\
    & \leq & C_\subs,
  \end{IEEEeqnarray*}
  where the last inequality follows from Theorem~\ref{the:diq2}.
  Therefore $C=C_\subs$.
\end{IEEEproof}

To verify the sufficient condition given in Theorem~\ref{the:non2}, we
do not need to check all input distributions that achieve $C$.  If
$p_{\lspan{X^\tr}}$ satisfies the sufficient condition in
Theorem~\ref{the:non2}, we can apply Lemma~\ref{cor:dd3} on
$p_{\lspan{X^\tr}}$ since $C=C_L$ for row-space-symmetric LOCs, and
obtain that $p_{\lspan{X^\tr}}'$ satisfies the sufficient condition
and has the structure defined in Lemma~\ref{cor:dd3}-1).
Therefore, we only need to check the sufficient condition for input
distributions with the structure that for each $r$,
$p_{\rank(X)}(r)=p_{\lspan{X^\tr}}(U_r)$ for certain $U_r\in \Gr(r,\ffield^M)$.

\begin{example}\label{example:rssopt}
  Since $\loc_{2}(H,1)$ is row-space-symmetric for any $H$
  (cf. Example~\ref{ex:rss}), we can use Theorem~\ref{the:non2} to
  characterize a sufficient condition such that $C=C_\subs$.
  
  Consider an input distribution $p^*$ with $p^*_{\rank(X)}(0)=p_0$
  and $p^*_{\rank(X)}(1)=p^*_{\lspan{X^\tr}}(U_1) = 1- p_0 = p_1$ for
  some $U_1\in \Gr(1,\ffield^M)$. Hence $p^*_{\lspan{X^\tr}}(U) = 0$
  for all $U\neq U_1 \in \Gr(1,\ffield^M)$.  We first check the
  sufficient condition.  By \eqref{eq:1ll}, for $p^*$, the Markov
  chain in the sufficient condition holds for any choices of $p_0$,
  $0\leq p_0 \leq 1$ and $U_1$. To satisfy the sufficient condition, we
  further require that $p^*$ achieves $C$. Now we assume $0< p_0 < 1$
  since otherwise, $p^*$ would not be capacity achieving unless the
  channel is trivial. A necessary and sufficient condition such that
  $p^*$ achieves $C$ is given by the KKT condition:
  \begin{IEEEeqnarray*}{rCl}
    C & = & \log \frac{1}{p_{\rank(Y)}(0)},  \\
    C & = & - \log p_1 + P_{\rank(Y)|\lspan{X^\tr}}(0|U_1) \log
    \frac{P_{\rank(Y)|\lspan{X^\tr}}(0|U_1)p_1}{p_{\rank(Y)}(0)},  \\
    C & \geq & - \log p_1 + P_{\rank(Y)|\lspan{X^\tr}}(0|U) \log
    \frac{P_{\rank(Y)|\lspan{X^\tr}}(0|U_1)p_1}{p_{\rank(Y)}(0)} \\
    & & + D_{\text{KL}}(P_{\lspan{Y^\tr}|\lspan{X^\tr}}(\cdot|U)||P_{\lspan{Y^\tr}|\lspan{X^\tr}}(\cdot|U_1)).
  \end{IEEEeqnarray*}
  Note that the first two equalities fix $p_1$ and $C$ as functions of
  $P_{\rank(Y)|\lspan{X^\tr}}(0|U_1)$. The third inequality gives a
  constraint for $U_1$, i.e., for all $U\neq U_1 \in
  \Gr(1,\ffield^M)$, we have
  \begin{IEEEeqnarray*}{rCl}
    D_{\text{KL}}(P(\cdot|U)||P(\cdot|U_1))
    & \leq & (P_{\rank(Y)|\lspan{X^\tr}}(0|U) -
  P_{\rank(Y)|\lspan{X^\tr}}(0|U_1)) 
     \log \frac {p_{\rank(Y)}(0)}{P_{\rank(Y)|\lspan{X^\tr}}(0|U_1)p_1}.
  \end{IEEEeqnarray*}
  Substituting the value of $p_1$, we have,
  \begin{IEEEeqnarray*}{rCl}
    \IEEEeqnarraymulticol{3}{l}{D_{\text{KL}}(P(\cdot|U)||P(\cdot|U_1))}
    \\
    & \leq & \frac{P_{\rank(Y)|\lspan{X^\tr}}(0|U) -
  P_{\rank(Y)|\lspan{X^\tr}}(0|U_1)}{P_{\rank(Y)|\lspan{X^\tr}}(1|U_1)}
 \log \frac{1}{P_{\rank(Y)|\lspan{X^\tr}}(0|U_1)},\ \forall U\neq U_1 \in
  \Gr(1,\ffield^M). \IEEEyesnumber \IEEEeqnarraynumspace \label{eq:sidsds}
  \end{IEEEeqnarray*}  
  Therefore, if \eqref{eq:sidsds} holds, there exists $p_1$ such that
  $p^*$ is capacity achieving, and hence $C=C_{\subs}$. 
\end{example}

\subsection{Degraded Linear Operator Channels}
\label{sec:34}

\begin{definition}\label{def:degraded}
  A LOC is \emph{degraded} if $\mutual(X;Y)=\mutual(\lspan{X};\lspan{Y})$ for all
$p_X$.
\end{definition}

By definition, it is clear that a degraded LOC has $C = C_{\subs}$.
Some degraded LOCs have been studied in the literature.
When $M=N$, the LOC with
$H$ uniformly distributed among all full rank $M\times M$ matrices is degraded
\cite{silva08c}.  If $H$ contains uniformly i.i.d. components, it was
shown that the corresponding LOC is also degraded\cite{siavoshani11}.  LOCs with
uniform-given-rank transfer matrices \cite{nobrega11, nobrega11a} are degraded, and
uniform-given-rank transfer matrices include the transfer matrices
studied in \cite{silva08c,siavoshani11} as special cases.

In this section, we focus on the general properties of degraded LOCs.
Since
\begin{IEEEeqnarray*}{rCl}
  \mutual(X;Y)
	  & = &\sum_{V,U\in \Pj(\ffield^T)} \sum_{\substack{\bX,\bY:\\\lspan{\bX}=U,\lspan{\bY}=V}} p(\bX,\bY) \log  \frac{p(\bX,\bY)}{p_X(\bX)p_{Y}(\bY)}  \\
	  & \geq &\sum_{V, U\in \Pj(\ffield^T)}
          p_{\lspan{X}\lspan{Y}}(U,V) \log
          \frac{p_{\lspan{X}\lspan{Y}}(U,V)}{p_{\lspan{X}}(U)p_{\lspan{Y}}(V)} \\
      & = &\mutual(\lspan{X};\lspan{Y}),
\end{IEEEeqnarray*}
where the inequality follows from the log-sum inequality (cf. \cite{yeung08}),
 a LOC is degraded if and only if
   \begin{IEEEeqnarray}{rCl}
     \label{eq:condoo1}
     \forall \bY,\ P_{Y|X}(\bY|\bX) & = & P_{Y|X}(\bY|\bX')\ \text{if}\
     \lspan{\bX}=\lspan{\bX'}, \\
     \noalign{\noindent and for all $p_X$\vspace{\jot}}
     \forall \bX, \
     \frac{P_{Y|X}(\bY|\bX)}{p_Y(\bY)}&=&\frac{P_{Y|X}(\bY'|\bX)}{p_Y(\bY')}
       \ \text{if}\
     \lspan{\bY}=\lspan{\bY'}. \IEEEeqnarraynumspace \label{eq:condoo2}
   \end{IEEEeqnarray}

\begin{example}\label{ex:ioosd}
We check when
$\loc_{2}(H,1)$ is degraded.
For this example, \eqref{eq:condoo2} holds trivially and
\eqref{eq:condoo1} is equivalent to
\begin{equation}\label{eq:iffloc21}
  P_{Y|X}(y|x_1) = P_{Y|X}(y|x_2), \ \forall y, x_1, x_2 \in
  \ffield^{1\times 2}.
\end{equation}
We have at most six linear constraints on the
distribution of $H$ such that $\loc_{2}(H,1)$ is degraded.

Note that \eqref{eq:iffloc21} is equivalent to \eqref{eq:cioss}. Hence
we can rephrase the conclusion of Example~\ref{example:uuu} as
$\loc_{2}(H,1)$ with a unique subspace degradation has $C=C_{\subs}$
if and only if it is degraded.  However, a LOC may not be degraded
even if $C=C_{\subs}$.  As an example, for the distribution of $H_2\in
\ffield^{2\times 2}$ in Table~\ref{tab:dist2}, $\loc_{2}(H_2,1)$ has
multiple subspace degradations, but $C=C_\subs=1$bit. The optimal
input distribution has $p_{\rank(X)}(0) = p_{\lspan{X^\tr}}(\lspan{[1\ 1]^\tr})=0.5$.
\end{example}

\begin{table}
   \centering
   \caption{A distribution over $\ffield_2^{2\times 2}$. Each numbered
     cell is the probability mass of the matrix whose first column is
     the row index of the table and second column is the column index of the table.}
   \label{tab:dist2}
   \begin{tabular}{c||c|c|c|c}
     & $\begin{bmatrix}0\\0\end{bmatrix}$
     & $\begin{bmatrix}1\\0\end{bmatrix}$
     & $\begin{bmatrix}0\\1\end{bmatrix}$
     & $\begin{bmatrix}1\\1\end{bmatrix}$ \\
     \hline \hline
     $\begin{bmatrix}0\\0\end{bmatrix}$ & 0 & 0 & $\frac{1}{6}$ & 0 \\
     \hline
     $\begin{bmatrix}1\\0\end{bmatrix}$ & 0 & 0 & $\frac{1}{12}$ &
     $\frac{1}{12}$ \\
     \hline
     $\begin{bmatrix}0\\1\end{bmatrix}$ & $\frac{1}{6}$ &
     $\frac{1}{12}$ & $\frac{1}{6}$ & $\frac{1}{12}$ \\
     \hline
     $\begin{bmatrix}1\\1\end{bmatrix}$ & 0 & $\frac{1}{12}$ &
     $\frac{1}{12}$ & 0 \\
   \end{tabular}
 \end{table}

\begin{theorem}\label{the:degradedrowspace}
  A degraded LOC has a unique
  subspace degradation and it is row-space-symmetric.
\end{theorem}
\begin{IEEEproof}
Fix a degraded LOC.
Since \eqref{eq:condoo1} implies the condition given in \eqref{eq:id90sl},
the subspace degradation is unique for a degraded LOC.
Fix full-row-rank $r\times M$ matrices $\mathbf{D}$ and $\mathbf{D}'$,
and an $r\times N$ matrix $\mathbf{E}$.
By \eqref{eq:condoo1}, for any full-column-rank matrix $\mathbf{B}$,
\begin{equation*}
  P_{Y|X}(\mathbf{BE}|\mathbf{BD}) = P_{Y|X}(\mathbf{BE}|\mathbf{BD}').
\end{equation*}
By Lemma~\ref{the:symm},
\begin{equation}
  \label{eq:2}
  \Pr\{\mathbf{D}H=\mathbf{E}\} = \Pr\{\mathbf{D}'H=\mathbf{E}\}.
\end{equation}
We show that the LOC is row-space symmetric using the above equality.

Fix any input $\bX$ and $\bX'$ and output $\bY$ and $\bY'$ satisfying $\lspan{\bY} \leq
\lspan{\bX}$, $\lspan{\bY'} \leq \lspan{\bX'}$, $\lspan{\bX^\tr} =
\lspan{\bX'^\tr}$ and $\lspan{\bY^\tr} = \lspan{\bY'^\tr}$.
Then we can write $\bX = \mathbf{BD}$, $\bY = \mathbf{BE}$,
$\bX'=\mathbf{B}'\mathbf{D}$ and $\bY'=\mathbf{B}'\mathbf{E}'$
(cf. \eqref{eq:ssod-1} and \eqref{eq:ssod-2} in
Appendix~\ref{sec:symmopt}).
Since $\lspan{\mathbf{E}^\tr} = \lspan{\mathbf{E}'^\tr}$, there exists
a full-rank square matrix $\Phi$ such that $\mathbf{E} = \Phi
\mathbf{E}'$.
Then we have
\begin{IEEEeqnarray}{rCl}
  P(\bY|\bX)
  & = & \label{eq:us1}
  \Pr\{\mathbf{D}H =\mathbf{E}\} \\
  & = & \IEEEnonumber
  \Pr\{\mathbf{D}H = \Phi\mathbf{E}'\} \\
  & = & \IEEEnonumber
  \Pr\{\Phi^{-1}\mathbf{D}H = \mathbf{E}'\} \\
  & = & \label{eq:us4}
  \Pr\{\mathbf{D}H =\mathbf{E}'\} \\
  & = & \label{eq:us5}
  P(\bY'|\bX'),
\end{IEEEeqnarray}
where \eqref{eq:us1} and \eqref{eq:us5} follow from
Lemma~\ref{the:symm}, and \eqref{eq:us4} follows from
\eqref{eq:2} and $\rank(\Phi^{-1}\mathbf{D}) = \rank(\mathbf{D})$.
The proof is completed by noting that \eqref{eq:us5} is sufficient for
a LOC being row-space-symmetric.
\end{IEEEproof}

The above theorem tells us that all the LOCs studied in
\cite{silva08c,siavoshani11, nobrega11, nobrega11a} have a unique
subspace degradation.  Now we have a better understanding of why the
capacity of these LOCs can be achieved by only optimizing
the input rank distribution.

We say a LOC is \emph{rank-symmetric} if its transition matrix is
rank-symmetric (see Definition~\ref{def:rsu2}). 
We see that  $\loc(H,T)$ is {rank-symmetric} if and only if there exists a function
  $\mu:\mathbb{Z}^+\times \mathbb{Z}^+\rightarrow [0\ 1]$ such that
\begin{equation*}
   P_{Y|X}(\bY|\bX) = \left\{ \begin{array}{ll}
    \mu(\rank(\bX),\rank(\bY)) &
    \lspan{\bY}\leq
    \lspan{\bX} \\ 0 & \text{otherwise}, \end{array}\right. \label{eq:uni0qq}
\end{equation*}
where $\mathbb{Z}^+$ is the set of nonnegative integers.

By the definition, we see that a rank-symmetric LOC is also row-space-symmetric
(cf. Definition~\ref{def:rsu} and Definition~\ref{def:rsu2}). The following theorem gives a
stronger characterization of rank-symmetric LOCs.

\begin{lemma}\label{lemma:uniform}
   A rank-symmetric LOC is degraded.
\end{lemma}
\begin{IEEEproof}%
   We can check that
   \eqref{eq:condoo1} and \eqref{eq:condoo2} hold for a rank-symmetric LOC.
   By the definition of rank-symmetric LOCs, we know that
   $P_{Y|X}(\bY|\bX)$ only depends on $U$ and $V$, which verifies \eqref{eq:condoo1}. By the same
   property of rank-symmetric LOCs,
  \begin{IEEEeqnarray*}{rCl}
    p_Y(\bY)
    & = &
    \sum_{\bX':V\leq \lspan{\bX'}} P_{Y|X}(\bY|\bX) p_X(\bX) \\
    & = &
    \sum_{U'\in \Pj(\ffield^T):V\leq
      U'} \mu(\dim(U'),\dim(V)) \sum_{\bX:\lspan{\bX}=U' }  p_{X}(\bX)
    \\
    & = &
    \sum_{r} \mu(r,\dim(V))
    \sum_{U'\in \Gr(r,\ffield^T):V\leq U'} p_{\lspan{X}}(U').
  \end{IEEEeqnarray*}
  This verifies \eqref{eq:condoo2}.
\end{IEEEproof}

But a degraded LOC may not be
rank-symmetric.

\begin{example}
  Consider $\loc_{2}(H_2,1)$, $H_2\in \ffield^{2\times 2}$ as an example. %
For the distribution of $H$ as given in Table~\ref{tab:dist1},
we can calculate that
\begin{IEEEeqnarray*}{rCl}
  P_{Y|X}(z_1|z_i) & = & \frac{1}{6}, \quad i=1,2,3 \\
  P_{Y|X}(z_2|z_i) & = & \frac{1}{6}, \quad i=1,2,3 \\
  P_{Y|X}(z_3|z_i) & = & \frac{1}{3}, \quad i=1,2,3 \\
  P_{Y|X}(z_0|z_i) & = & \frac{1}{3}, \quad i=1,2,3,
\end{IEEEeqnarray*}
   where 
\begin{equation}\label{eq:z}
  z_0  =  [0\ 0], \quad z_1  =  [1\ 0], \quad
  z_2 =  [0\ 1], \ \text{and}\  z_3  =  [1\ 1].
\end{equation}
We can check by \eqref{eq:iffloc21} that $\loc_{2}(H,1)$ with the distribution of $H$ given in
Table~\ref{tab:dist1} is degraded. But this LOC is not rank symmetric.
\end{example}

 \begin{table}
   \centering
   \caption{A distribution over $\ffield_2^{2\times 2}$. Each numbered
     cell is the probability mass of the matrix whose first column is
     the row index of the table and second column is the column index of the table.}
   \label{tab:dist1}
   \begin{tabular}{c||c|c|c|c}
     & $\begin{bmatrix}0\\0\end{bmatrix}$
     & $\begin{bmatrix}1\\0\end{bmatrix}$
     & $\begin{bmatrix}0\\1\end{bmatrix}$
     & $\begin{bmatrix}1\\1\end{bmatrix}$ \\
     \hline \hline
     $\begin{bmatrix}0\\0\end{bmatrix}$ & 0 & $\frac{1}{12}$ & $\frac{1}{12}$ & $\frac{1}{12}$ \\
     \hline
     $\begin{bmatrix}1\\0\end{bmatrix}$ & $\frac{1}{12}$ & $\frac{1}{6}$ & 0 &
     0 \\
     \hline
     $\begin{bmatrix}0\\1\end{bmatrix}$ & $\frac{1}{12}$ &
     0 & $\frac{1}{6}$ & 0 \\
     \hline
     $\begin{bmatrix}1\\1\end{bmatrix}$ & $\frac{1}{12}$ &  0 &
     0 & $\frac{1}{6}$ \\
   \end{tabular}
 \end{table}

The following theorem shows the relation between uniform-given-rank
transfer matrices and rank-symmetric LOCs.

\begin{theorem}\label{the:uniformeq}
  Let $H$ be a random matrix with dimension $M\times N$.
  i) If $T\geq M$ and  $\loc(H,T)$ is rank-symmetric,
      then $H$ is uniform-given-rank.
  ii) If $H$ is uniform-given-rank, then $\loc(H,T)$ is rank-symmetric.
\end{theorem}
\begin{IEEEproof}%
  Proof of i).
  Fix $\bX\in
  \ffield^{T\times M}$ with $\rank(\bX)=M$. The existence of
  such $\bX$ follows from $T\geq M$. For any $\bY\in
  \ffield^{T\times N}$, we have a unique $\bH$ such that
  $\bY=\bX\bH$. Since the LOC is rank-symmetric,
  \begin{align*}
    p_{H}(\bH) & = \Pr\{\bY = \bX H\} \\
    & = \mu(\rank(\bX),\rank(\bY)) \\
    & = \mu(M,\rank(\bH)).
  \end{align*}
  Therefore $H$ is uniform-given-rank.

  Proof of ii). Fix $\bX\in \ffield^{T\times M}$ and $\bY\in
  \ffield^{T\times N}$ with $\rank(\bX)=r$, $\rank(\bY)=s$ and
  $\lspan{\bY}\leq \lspan{\bX}$.
  By the similar procedure for obtaining \eqref{eq:yam3}, we have
  \begin{equation*}
    P_{Y|X}(\bY|\bX) = \Pr\left\{   \mathbf{D} H = \begin{bmatrix} \mathbf{E} \\
      \mathbf{0} \end{bmatrix} \right\},
  \end{equation*}
  for certain full row-rank matrices $\mathbf{D}$ and $\mathbf{E}$ are full row-rank matrices satisfying
  $\lspan{\mathbf{D}^\tr} = \lspan{\bX^\tr}$ and $\lspan{\mathbf{E}^\tr} =
  \lspan{\bY^\tr}$.
  Fix any full-row-rank matrices $\mathbf{D}'\in\ffield^{r\times M}$ and
  $\mathbf{E}'\in\ffield^{s\times N}$.
  Find full rank matrices $\Phi$ and $\Psi$ such that $\mathbf{D}'
  = \mathbf{D}\Phi$ and $\mathbf{E}' = \mathbf{E}\Psi$.
  We have
  \begin{equation*}
    \Pr\left\{   \mathbf{D'} H = \begin{bmatrix} \mathbf{E'} \\
      \mathbf{0} \end{bmatrix} \right\}
     =
    \Pr\left\{   \mathbf{D} \Phi H \Psi^{-1} = \begin{bmatrix} \mathbf{E} \\
      \mathbf{0} \end{bmatrix} \right\} 
     = 
    \Pr\left\{   \mathbf{D} H = \begin{bmatrix} \mathbf{E} \\
      \mathbf{0} \end{bmatrix} \right\},
  \end{equation*}
  where the last equality follows that $H$ is uniform-given-rank. Hence
  \begin{equation*}
    P_{Y|X}(\bY|\bX) = \Pr\left\{   \mathbf{D'} H = \begin{bmatrix} \mathbf{E'} \\
      \mathbf{0} \end{bmatrix} \right\}.
  \end{equation*}
  So $P_{Y|X}(\bY|\bX)$ only relates to the ranks of $\bX$ and $\bY$,
  i.e., $\loc(H,T)$ is rank-symmetric.
\end{IEEEproof}

There exists
rank-symmetric LOCs with non-uniform-given-rank transfer matrices.

\begin{example}
  We give an example of a rank-symmetric LOC
that has a non-uniform transfer matrix.
Consider $\loc_{2}(H_2,1)$, $H_2\in \ffield^{2\times 2}$ with
\begin{equation*}
  p_{H_2}(\mathbf{H})= \frac{1}{4}, \ \text{for}\ \mathbf{H} =
  \begin{bmatrix} 1 & 0 \\ 0 & 0 \end{bmatrix},
  \begin{bmatrix} 0 & 1 \\ 1 & 0 \end{bmatrix},
  \begin{bmatrix} 0 & 0 \\ 0 & 1 \end{bmatrix},
  \begin{bmatrix} 1 & 1 \\ 1 & 1 \end{bmatrix}.
\end{equation*}
We can check that $H_2$ is not uniform-given-rank, but we
can verify that $\loc_{2}(H,1)$ is rank-symmetric.
\end{example}

In the last of this section, we verify a claim given in
Section~\ref{sec:boundsss}.  Note that any transfer matrix $H$ can be
converted to a uniform-given-rank transfer matrix $H^*$ with the same
rank distribution \cite{nobrega11, nobrega11a, siavo12avc} obtained by
$H^* = \Phi H \Psi$, where $\Phi$ and $\Psi$ are independent uniformly
distributed random matrices in $\Fr(\ffield^{M\times M})$ and
$\Fr(\ffield^{N\times N})$, respectively. Hence $C(H,T) \geq C(H^*,T)
= C_{\subs}(H^*,T)$. We show that the lower bound on $C(H,T)$ given in
Theorem~\ref{the:89s} is at least as good as $C_{\subs}(H^*,T)$.

Let $p_X$ be a uniform-given-rank input distribution that achieves
$C_{\subs}(H^*,T)$, the existence of such a distribution is guaranteed
by Theorem~\ref{the:uniquesd}. Let $P^*$ and $P$ be the transition
matrices corresponding to $H^*$ and $H$ respectively. For any input
matrix $\bX'$ with $\rank(\bX') = r$, we have 
\begin{IEEEeqnarray*}{rCl}
  P^*_{\rank(Y)|\rank(X)}(s|r) 
  & = & 
  P^*_{\rank(Y)|X}(s|\bX') \\
  & = &
  \Pr\{ \rank(\bX' \Phi H \Psi) = s\} \\
  & = &
  \Pr\{ \rank(\bX' \Phi H) = s\} \\ 
  & = & 
  \sum_{\bX:\rank(\bX)=r} \Pr\{ \rank(\tilde X H) = s, \tilde X = \bX \} \\
  & = &
  \sum_{\bX:\rank(\bX)=r} P_{\rank(Y)|X}(s|\bX) p_{X|\rank(X)}(\bX|r)
  \\
  & = & 
  P_{\rank(Y)|\rank(X)}(s|r),
\end{IEEEeqnarray*}
where the first equality follows from \eqref{eq:di8s}; and $\tilde X = \bX'\Phi$ is uniformly distributed among all input matrices with
rank $r$, and has the same distribution of $p_{X|\rank(r)}(\bX|r)$.
Therefore for $p_X$,
\begin{IEEEeqnarray*}{rCl}
  I(X;Y)|_{p_H}  & \geq & 
  \locrate(\rank(X),\rank(Y))|_{P_{\rank(Y)|\rank(X)}} 
   + I(\lspan{X^\tr};\lspan{Y^\tr})|_{P_{\lspan{Y^\tr}|\lspan{X^\tr}}} \\
  & \geq & 
  \locrate(\rank(X),\rank(Y))|_{P_{\rank(Y)|\rank(X)}}
   + I(\rank(X);\rank(Y))|_{P_{\rank(Y)|\rank(X)}} \\
  & = & 
  \locrate(\rank(X),\rank(Y))|_{P_{\rank(Y)|\rank(X)}^*} 
   + I(\rank(X);\rank(Y))|_{P_{\rank(Y)|\rank(X)}^*} \\
  & = & 
  C_{\subs}(H^*,T),
\end{IEEEeqnarray*}
where the first inequality follows from Theorem~\ref{the:89s} and the
last equality follows from Theorem~\ref{the:uniquesd}. Thus $C(H,T)
\geq C_{\subs}(H^*,T)$.

\section{Concluding Remarks}
\label{sec:con}

In this paper, we studied upper and lower bounds for both the Shannon
capacity and the subspace coding capacity of LOCs. We characterized
various classes of LOCs with different properties of these bounds, where row spaces and ranks of input and output matrices play important roles.

Our results provide some guidelines for coding design.
Subspace coding is good for LOCs with a unique subspace degradation
since otherwise we have difficulty to find an optimal input
distribution for subspace coding.  For general LOCs, we can use
constant-rank uniform-given-row-space input distribution for subspace
codes since 1) such an optimal input distribution is relatively easy
to compute, and 2) the loss of rate, compared with the subspace coding
capacity, can be small for typical parameters.  Further, it is not
always optimal to uniform input and output of a LOC for applying
subspace coding.

We are motivated to consider other coding schemes for LOCs since
for many cases either the optimal subspace coding scheme is difficult
to find or subspace coding is not capacity achieving.  Readers are
referred to \cite{yang12sumas} for a superposition based coding
scheme that can achieve rate higher than subspace coding.

\appendices

\section{Symmetry Properties in Channel Capacity
  Optimization}
\label{sec:symmopt}

We discuss how the symmetry properties is used to solve the optimization problem for finding the channel capacity of LOCs.
This is useful for getting some numerical results.
We first introduce some notations that will be used in this section
and Appendix~\ref{sec:pf-claims}.

Let $\mathbf B$ be a $t\times r$ matrix with rank $r$, i.e., $\mathbf
B$ is of full column-rank. For a $t\times
m$ matrix $\mathbf A$ with $\lspan{\mathbf A}\leq \lspan{\mathbf
  B}$, define $\mathbf A/\mathbf B$ to be a matrix such that $\mathbf A
= \mathbf B(\mathbf A/\mathbf B)$.  The notation ``$/$'' is well
defined because i) there always exists $\mathbf C$ such that $\mathbf
A = \mathbf{BC}$ since $\lspan{\mathbf A}\leq \lspan{\mathbf B}$
and ii) such $\mathbf C$ is unique since $\mathbf B$ is full column
rank.

Let $\bX$ and $\bY$ be the input and output matrices of $\loc(H,T)$,
respectively, with $\lspan{\bY}\leq \lspan{\bX}$.  A decomposition
of $\bX$ and $\bY$ as in Lemma~\ref{the:symm} can be found as follows.
First, fix a full column rank matrix $\mathbf B$ with $\lspan{\bX} =
\lspan{\mathbf B}$. Then, $\bX = \mathbf B(\bX/\mathbf B)$ and $\bY = \mathbf B
  (\bY/\mathbf B)$. By  Lemma~\ref{the:symm},
\begin{equation*}%
  P_{Y|X}(\bY|\bX)  = \Pr\{ (\bX/\mathbf B) H = \bY/\mathbf B \}.
\end{equation*}

For $U \in
  \Pj(\ffield^M)$, let $\mathbf{D}_U$ be a $\dim(U)\times M$ matrix
  with $\lspan{\mathbf{D}_U^\tr} = U$.  For any $\bX\in
  \ffield^{T\times M}$ and $\bY\in \ffield^{T\times N}$ with
  $\lspan{\bY}\leq \lspan{\bX}$, we can write
 \begin{IEEEeqnarray}{rCl}
   \bX & = & \label{eq:ssod-1}
   \mathbf{BD}_{\lspan{\bX^\tr}}, \\
   \bY & = & \label{eq:ssod-2}
   \mathbf{BE},
 \end{IEEEeqnarray}
 where $\mathbf{B}^T = \bX^\tr/\mathbf{D}_{\lspan{\bX^\tr}}^{\tr}$ and
 $\mathbf{E} = \bY/\mathbf{B}$.

Due to the symmetry properties of the matrix of transition
probabilities in Lemma~\ref{prop:1}, it is not necessary to
calculate $P_{Y|X}(\bY|\bX)$ for all pairs of $\bX$ and $\bY$.
 For each subspace $U \in \Pj(\min\{T,M\},
\ffield^{M})$, we choose one full row rank matrix $\mathbf D$ with
$\lspan{\mathbf D^\tr} = U$ to compute $(\Pr\{\mathbf D H = \mathbf
E\}: \mathbf E\in \ffield^{k\times N})$.  Then for any $\bX$ and $\bY$
with $\lspan{\bY} \leq \lspan{\bX}$ and $\lspan{\bX^\tr} = U$, we
know $P_{Y|X}(\bY|\bX) = \Pr\{\mathbf D H = \mathbf
\bY/(\bX^\tr/\mathbf{D}^\tr)^\tr\}$.  The overall complexity of
computing the transition matrix is
\begin{equation*}
  \sum_{k=0}^{\min\{T,M\}} \gcos{M}{k} q^{kN}
  < \left\{\begin{array}{ll}
      cq^{MN} & M\leq \min\{T,N\} \\
      c'q^{L(M+N-L)} & \text{otherwise},
    \end{array} \right.
\end{equation*}
where $L=\min\{T,(M+N)/2\}$, $c$ and $c'$ are constants. The
inequality for $M\leq \min\{T,N\}$ is obtained as
follows\footnote{This method is suggested by an anonymous reviewer,
  who is thereby acknowledged.}.  We have
\begin{IEEEeqnarray*}{rCl}
  \cmat{N}{k} & = & q^{Nk}\prod_{i=N-k+1}^N (1-q^{-i}) \\
  & > & q^{Nk} \prod_{i=1}^\infty (1-q^{-i}) \\
  & \geq & \kappa q^{Nk}, \IEEEyesnumber \label{eq:kappa}
\end{IEEEeqnarray*}
where $\kappa = \prod_{i=0}^\infty (1- 2^{-i}) \approx 0.28879$ is a constant \cite{cooper00}.
Thus
\begin{IEEEeqnarray*}{rCl}
  \sum_{k=0}^{M} \gcos{M}{k} q^{kN}
  & < & 1/\kappa \sum_{k=0}^{M}\gcos{M}{k} \cmat{N}{k} \\
  & = & 1/\kappa \sum_{k=0} \cmat{M,N}{k} \\
  & = & 1/\kappa q^{MN},
\end{IEEEeqnarray*}
where the last equality is obtained by \eqref{eq:ckss66}.  When $M > \min\{T,N\}$, we have
\begin{IEEEeqnarray*}{rCl}
  \sum_{k=0}^{\min\{T,M\}} \gcos{M}{k} q^{kN}
  & = & \sum_{k=0}^{\min\{T,M\}} \frac{\cmat{M}{k}}{\cmat{k}{k}} q^{kN} \\
  & < & \sum_{k=0}^{\min\{T,M\}} \frac{1/\kappa q^{Mk}}{q^{k^2}} q^{kN} \\
  & = & 1/\kappa\sum_{k=0}^{\min\{T,M\}} q^{k(M+N-k)}, \IEEEyesnumber \label{eq:apb}
\end{IEEEeqnarray*}
where the inequality is obtained by \eqref{eq:kappa} with $M$ in place
of $N$ and $\cmat{k}{k} < q^{k^2}$. Note that the $k(M+N-k)$ in
\eqref{eq:apb} takes its maximum at $k=\min\{\min\{T,M\},(M+N)/2\}=L$.
Then by a technique similar to the one used in
\cite[Lemma 1]{gadouleau10}, the inequality for $M > \min\{T,N\}$ is
obtained, where the constant $c' = 1/\kappa \sum_{i=0}^\infty
2^{-i^2} \approx 5.4137$.

After obtaining the transition matrix, we can find an optimal
input distribution by solving the maximization problem in
Theorem~\ref{the:diq}, which is equivalent to
finding an optimal distribution over
$\Pj(\min\{T,M\},\ffield^{M})$.
Since
$|\Pj(\min\{M,T\},\ffield^{M})| = \sum_{k=0}^{\min\{M,T\}} \gco{M}{k}$,
we can bound the number of probability masses to determine as
\begin{equation*}
  \sum_{k=0}^{\min\{M,T\}} \gcos{M}{k} <
  \left\{
      \begin{array}{ll}
        \Theta_1q^{M^2/4} &  T\geq M/2 \\
        \Theta_2q^{T(M-T)} & \text{otherwise},
      \end{array}
   \right.
\end{equation*}
where $\Theta_1$ and $\Theta_2$ are constants. The inequality for
$T\geq M/2$ is obtained by \cite[Lemma 1]{gadouleau10}, while the
inequality for $T< M/2$, is obtained by \cite[Proposition 1]{gadouleau10}.

\section{Proof of Claims in the Proof of Theorem~\ref{the:89s}}
\label{sec:pf-claims}

\begin{IEEEproof}[Proof of Claim~\ref{claim:1}]
Define random variable $Y^{(r)} = Y$ for $r = 0$. For $r = 1, \ldots, \min\{T,M\}$,
define random variables $Y^{(r)}$ and $Y^{(r,\Phi)}$ with $\Phi\in \Fr(\ffield^{r\times r})$ over $\ffield^{T\times N}$ as follows.
For $\lspan{\bY}\leq \lspan{\bX}$, let
\begin{equation*} %
    P_{Y^{(r,\Phi)}|X}(\bY | \bX)  = \left \{
    \begin{array}{ll}
      P_{Y^{(r-1)}|X}(\bY | \bX) & \rank(\bX)\neq r, \\
      \Pr\{\mathbf{D}_{\lspan{\bX^\tr}} H = \Phi \mathbf{E} \} &
      \rank(\bX)=r,
    \end{array}
    \right.
\end{equation*}
where $\mathbf{E} =
\bY/(\bX^\tr/\mathbf{D}_{\lspan{\bX^\tr}}^{\tr})^{\tr}$ (cf.~\eqref{eq:ssod-1}
and \eqref{eq:ssod-2}). Random variables $Y^{(r)}$ are over $\ffield^{T\times N}$ such that for $\lspan{\bY}\leq \lspan{\bX}$,
\begin{equation*}%
  P_{Y^{(r)}|X}(\bY|\bX) = \frac{1}{\cmat{r}{r}} \sum_{\Phi\in\Fr(\ffield^{r\times
  r})} P_{Y^{(r,\Phi)}|X}(\bY | \bX).
\end{equation*}
Note that when $\rank(\bX) > r$,
\begin{equation}\label{eq:ccis-0}
  P_{Y^{(r,\Phi)}|X}(\bY | \bX)  = P_{Y^{(r)}|X}(\bY | \bX)  = P_{Y|X}(\bY|\bX).
\end{equation}

We will show that for $r=1,\ldots,\min\{M,T\}$,
\begin{equation}
  \label{eq:claim11id}
  P_{Y^*|X} = P_{Y^{(\min\{M,T\})}|X},
\end{equation}
and
\begin{equation}
  \label{eq:claim11ls}
  \mutual(X;Y^{(r,\Phi)}) = \mutual(X;Y^{(r-1)}).
\end{equation}
Since for a fixed $p_X$, mutual information $I(X;Y)$ is a convex
function of the transition probabilities, we have
\begin{equation*}
  \mutual(X;Y^{(r)}) \leq  \frac{1}{\cmat{r}{r}} \sum_{\Phi\in\Fr(\ffield^{r\times
  r})} \mutual(X;Y^{(r,\Phi)})  = \mutual(X;Y^{(r-1)}).
\end{equation*}
Then, the lemma is proved by
\begin{equation*}
  \mutual(X;Y) = \mutual(X;Y^{(0)}) \geq  \mutual(X;Y^{(\min\{M,T\})}) = \mutual(X;Y^*).
\end{equation*}

We first prove \eqref{eq:claim11id}.
For $\bX$ and $\bY$ with $\rank(\bX)=r$, $\rank(\bY)=s$, $\lspan{\bY^\tr}=V$, $\lspan{\bX^\tr}=U$ and $\lspan{\bY}\leq \lspan{\bX}$,
by definition,
  \begin{equation*}
  P_{Y^{(\min\{M,T\})}|X}(\bY|\bX) =
    \frac{1}{\cmat{r}{r}} \sum_{\Phi \in \Fr(\ffield^{r \times r})} \Pr\{\mathbf{D}_{U} H = \Phi \mathbf{E} \},
  \end{equation*}
  where $\mathbf{D}_U$ and $\mathbf{E}$ are defined in
  \eqref{eq:ssod-1} and \eqref{eq:ssod-2}, respectively.
For $\mathbf{E}_0 \in \mathcal{E} \triangleq
\{\mathbf{K}\in \ffield^{r\times N}:\lspan{\mathbf{K}^{\tr}}
=V\}$, let $\mathcal{C}(\mathbf{E}_0) =
\{\mathbf{C}\in \Fr(\ffield^{r\times r}):\mathbf{CE} = \mathbf{E}_0\}$.  We
see that $\{\mathcal{C}(\mathbf{E}_0), \mathbf{E}_0\in \mathcal{E}\}$
gives a partition of $\Fr(\ffield^{r\times r})$. Since $\mathcal{C}(\mathbf{E}_0)$ for all $\mathbf{E}_0\in \mathcal{E}$
have the same cardinality, $|\mathcal{C}(\mathbf{E}_0)| =
\frac{|\Fr(\ffield^{r\times r})|}{|\mathcal{E}|} =
\frac{\cmat{r}{r}}{\cmat{r}{s}}$ for all $\mathbf{E}_0\in
\mathcal{E}$. Therefore,
\begin{IEEEeqnarray*}{rCl}
  \frac{1}{\cmat{r}{r}} \sum_{\Phi \in
      \Fr(\ffield^{r \times r})} \Pr\{\mathbf{D}_{U} H = \Phi
    \mathbf{E} \}
  & = &
  \frac{1}{\cmat{r}{r}} \sum_{\mathbf{E}_0 \in \mathcal{E}} \sum_{\Phi \in \mathcal{C}(\mathbf{E}_0)} \Pr\{\mathbf{D}_{U} H = \Phi \mathbf{E}\} \\
  & = &
  \frac{1}{\cmat{r}{r}} \sum_{\mathbf{E}_0 \in \mathcal{E}} |\mathcal{C}(\mathbf{E}_0)| \Pr\{\mathbf{D}_{U} H = \mathbf{E}_0\} \\
  & = &
  \frac{1}{\cmat{r}{s}} \sum_{\mathbf{E}_0 \in \mathcal{E}} \Pr\{\mathbf{D}_{U} H = \mathbf{E}_0\} \\
  & = &
  \frac{1}{\cmat{r}{s}} \Pr\{\lspan{(\mathbf{D}_{U} H)^\tr} = V\}\\
  & = &
  \frac{1}{\cmat{r}{s}} P_{\lspan{Y^{\tr}}|\lspan{X^{\tr}}}(V|U).
\end{IEEEeqnarray*}
By the
definition of $Y^*$, \eqref{eq:claim11id} is proved.

Now we prove \eqref{eq:claim11ls}. First, we have for $i\neq r$,
\begin{IEEEeqnarray*}{rCl}
  p_{\rank(X),Y^{(r,\Phi)}} (i,\bY)
  & = &
    \sum_{\bX:\rank(\bX)=i} P_{Y^{(r,\Phi)}|X}(\bY|\bX)p_{X}(\bX) \\
  & = & %
    \sum_{\bX:\rank(\bX)=i} P_{Y^{(r-1)}|X}(\bY|\bX)p_{X}(\bX) \\
  & = & \IEEEyesnumber \label{eq:ccis-11}
  p_{\rank(X),Y^{(r-1)}} (i,\bY),
\end{IEEEeqnarray*}
where the second equality is obtained by the definition of
$P_{Y^{(r,\Phi)}|X}(\bY|\bX)$ for $\rank(\bX)\neq r$.
Specifically, when $r < i$,
\begin{equation*}
  p_{\rank(X),Y^{(r)}} (i,\bY)
   = p_{\rank(X),Y^{(r-1)}} (i,\bY)
\end{equation*}
by the definition of $P_{Y^{(r)}|X}$ and \eqref{eq:ccis-11}.
Recursively applying the above equality, we have that when $r < i$,
\begin{equation}
  \label{eq:ccis-12}
  p_{\rank(X),Y^{(r)}} (i,\bY) = p_{\rank(X),Y} (i,\bY).
\end{equation}

We also have
\begin{IEEEeqnarray*}{rCl}
  p_{\rank(X),Y^{(r,\Phi)}} (r,\bY)
  & = &
  \sum_{\bX:\rank(\bX)=r} P_{Y^{(r,\Phi)}|X}(\bY|\bX)p_{X}(\bX) \\
  & = &
  \sum_{U\in \Gr(r,\ffield^M)} \sum_{\mathbf{B}\in
    \Fr(\ffield^{T\times r})} P_{Y^{(r,\Phi)}|X}(\bY|\mathbf{BD}_{U})p_{X}(\mathbf{BD}_U) \\
  & = & %
  \sum_{U\in \Gr(r,\ffield^M)} \sum_{\mathbf{B}\in
    \Fr(\ffield^{T\times r})} \Pr\{\mathbf{D}_{U} H = \Phi
  (\bY/\mathbf{B}) \} \frac{p_{\lspan{X^{\tr}}}(U)}{\cmat{T}{r}} \\
  & = & \sum_{U\in \Gr(r,\ffield^M)}
  \frac{p_{\lspan{X^{\tr}}}(U)}{\cmat{T}{r}}
  \sum_{\mathbf{B}'\in
    \Fr(\ffield^{T\times r})} \Pr\{\mathbf{D}_{U} H =
  \bY/\mathbf{B}' \} \\
  & = &
  p_{\rank(X),Y} (r,\bY) \\
  & = & \IEEEyesnumber \label{eq:ccis-31}
  p_{\rank(X),Y^{(r-1)}} (r,\bY),
\end{IEEEeqnarray*}
where the third equality follows that $p_X$ is uniform-given-row-space and the
definition of $P_{Y^{(r,\Phi)}|X}(\bY|\bX)$ for $\rank(\bX) = r$;
the forth equality follows by $\Phi (\bY/\mathbf{B}) =
\bY/(\mathbf{B}\Phi^{-1})$ and substituting $\mathbf{B}\Phi^{-1}$ by
$\mathbf{B}'\in \Fr(\ffield^{T\times r})$; and \eqref{eq:ccis-31}
follow from \eqref{eq:ccis-12}.

Thus, by \eqref{eq:ccis-11} and \eqref{eq:ccis-31}, we have
\begin{IEEEeqnarray*}{rCl}
  p_{Y^{(r,\Phi)}}(\bY)
  & = &
  \sum_{i}p_{\rank(X),Y^{(r,\Phi)}}(i,\bY) \\
  & = &
  \sum_{i}p_{\rank(X),Y^{(r-1)}}(i,\bY) \\
  & = &
  p_{Y^{(r-1)}}(\bY),
\end{IEEEeqnarray*}
 and hence
\begin{equation}\label{eq:ccis-7}
  \entropy(Y^{(r,\Phi)})=\entropy(Y^{(r-1)}).
\end{equation}

Further, for $\bX$ with $\rank(\bX)\neq r$, since
$P_{Y^{(r,\Phi)}|X}(\bY | \bX)  =   P_{Y^{(r-1)}|X}(\bY | \bX)$, we have
\begin{equation}\label{eq:ccis-5}
  \entropy(Y^{(r,\Phi)}|X = \bX)
   =
   \entropy(Y^{(r-1)}|X = \bX).
\end{equation}
On the other hand, for $\bX$ with $\rank(\bX) = r$, by substituting $\bX=\mathbf{BD}_{\lspan{\bX^\tr}}$, we have
\begin{IEEEeqnarray*}{rCl}
  \entropy(Y^{(r,\Phi)}|X = \bX)
  & = & \IEEEyesnumber \label{eq:ccis-41}
  \sum_{\bY:\lspan{\bY}\leq \lspan{\mathbf{B}}}
  \Pr\{\mathbf{D}_{\lspan{\bX^\tr}} H = \Phi (\bY/\mathbf{B}) \} 
   \log
  \frac{1}{\Pr\{\mathbf{D}_{\lspan{\bX^\tr}} H = \Phi (\bY/\mathbf{B}) \}}
  \\
  & = & \IEEEyesnumber \label{eq:ccis-42}
  \sum_{\bY:\lspan{\bY}\leq \lspan{\mathbf{B}}}
  \Pr\{\mathbf{D}_{\lspan{\bX^\tr}} H = \bY/(\mathbf{B}\Phi^{-1}) \}  \log
  \frac{1}{\Pr\{\mathbf{D}_{\lspan{\bX^\tr}} H = \bY/(\mathbf{B}\Phi^{-1}) \}}
  \\
  & = & \IEEEyesnumber \label{eq:ccis-43} 
  \sum_{\bY:\lspan{\bY}\leq \lspan{\mathbf{B}}}
  P_{Y|X}(Y|\mathbf{B}\Phi^{-1}\mathbf{D}_{\lspan{\bX^\tr}}) \log
  \frac{1}{P_{Y|X}(Y|\mathbf{B}\Phi^{-1}\mathbf{D}_{\lspan{\bX^\tr}})}
  \\
  & = & \IEEEyesnumber \label{eq:ccis-44}
  \entropy(Y|X = \mathbf{B}\Phi^{-1}\mathbf{D}_{\lspan{\bX^\tr}})
\end{IEEEeqnarray*}
where \eqref{eq:ccis-41} follows from $\lspan{\bX} = \lspan{\mathbf{B}}$
  and the definition of $P_{Y^{(r,\Phi)}|X}(\bY|\bX)$ for $\rank(\bX) = r$;
\eqref{eq:ccis-42} follows by $\Phi (\bY/\mathbf{B}) =
\bY/(\mathbf{B}\Phi^{-1})$; \eqref{eq:ccis-43} follows from Lemma~\ref{the:symm}; and
\eqref{eq:ccis-44} is obtained by $\lspan{\mathbf{B}\Phi^{-1}\mathbf{D}_{\lspan{\bX^\tr}}} = \lspan{\mathbf{B}}$.

Hence
\begin{IEEEeqnarray*}{rCl}
  \sum_{\bX:\rank(\bX)=r}
    \entropy(Y^{(r,\Phi)}|X=\bX) p_{X}(\bX)
  & = &
  \sum_{U\in \Gr(r,\ffield^M)} \sum_{\mathbf{B}\in \Fr(\ffield^{T\times
      r})} \entropy(Y^{(r,\Phi)}|X=\mathbf{B}\mathbf{D}_{U}) p_{X}(\mathbf{B}\mathbf{D}_{U})  \\
  & = &
  \sum_{U\in \Gr(r,\ffield^M)} \sum_{\mathbf{B}\in \Fr(\ffield^{T\times
      r})} \entropy(Y|X=\mathbf{B}\Phi^{-1}\mathbf{D}_{U}) p_{X}(\mathbf{B} \mathbf{D}_{U}) \\
  & = &
  \sum_{U\in \Gr(r,\ffield^M)} \sum_{\mathbf{B}'\in \Fr(\ffield^{T\times
      r})} \entropy(Y|X=\mathbf{B}'\mathbf{D}_{U}) p_{X}(\mathbf{B}'
  \mathbf{D}_{U}) \\
  & = &
  \sum_{\bX:\rank(\bX)=r}\entropy(Y|X=\bX) p_{X}(\bX) \\
  & = & \IEEEyesnumber \label{eq:85w}
  \sum_{\bX:\rank(\bX)=r}\entropy(Y^{(r-1)}|X=\bX) p_{X}(\bX),
\end{IEEEeqnarray*}
where the second equality follows from \eqref{eq:ccis-44}; the third
equality follows by substituting $\mathbf{B}\Phi^{-1}$ by
$\mathbf{B}'\in \Fr(\ffield^{T\times r})$ and the fact that $p_{X}$ is
uniform-given-row-space; and \eqref{eq:85w} follows from \eqref{eq:ccis-0}.

Therefore,
\begin{IEEEeqnarray*}{rCl}
  \entropy(Y^{(r,\Phi)}|X)
  & = &
  \sum_{\bX:\rank(\bX)\neq r} \entropy(Y^{(r,\Phi)}|X=\bX) p_{X}(\bX) 
  + \sum_{\bX:\rank(\bX)=r} \entropy(Y^{(r,\Phi)}|X=\bX) p_{X}(\bX) \\
  & = & \IEEEyesnumber \label{eq:ccis-6}
  \sum_{\bX:\rank(\bX)\neq r} \entropy(Y^{(r-1)}|X=\bX) p_{X}(\bX) 
  + \sum_{\bX:\rank(\bX)=r} \entropy(Y^{(r-1)}|X=\bX) p_{X}(\bX) \\
  & = &  \IEEEyesnumber \label{eq:ccis-62}
  \entropy(Y^{(r-1)}|X),
\end{IEEEeqnarray*}
where  \eqref{eq:ccis-6} follows from \eqref{eq:ccis-5} and
\eqref{eq:85w}.
Lastly, the equality in \eqref{eq:claim11ls} is proved by \eqref{eq:ccis-7} and \eqref{eq:ccis-62}.
\end{IEEEproof}

\begin{IEEEproof}[Proof of Claim~\ref{claim:2}]
By the definition of $P_{Y^*|X}$,
\begin{IEEEeqnarray*}{rCl}
  \entropy(Y^*|X) & = &
  \sum_{\bX} p_X(\bX) \sum_{\bY:\lspan{\bY}\leq\lspan{\bX}} P_{Y^*|X}(\bY|\bX) \log \frac{1}{P_{Y^*|X}(\bY|\bX)} \\
  & = &
  \sum_r \sum_{U\in \Gr(r,\ffield^M)}
  \sum_{\bX:\lspan{\bX^\tr}=U}
  \frac{p_{\lspan{X^\tr}}(U)}{\cmat{T}{r}}
  \sum_s \sum_{V\in \Gr(s,\ffield^N)} \\
  & & \times
  \sum_{\bY:\lspan{\bY^\tr}=V,\lspan{\bY}\leq\lspan{\bX}}
  \frac{P_{\lspan{Y^\tr}|\lspan{X^\tr}}(V|U)}{\cmat{r}{s}}
  \log \frac{\cmat{r}{s}}{P_{\lspan{Y^\tr}|\lspan{X^\tr}}(V|U)} \\
  & = &
  \sum_r \sum_{U\in \Gr(r,\ffield^M)}
  p_{\lspan{X^\tr}}(U)
  \sum_s \sum_{V\in \Gr(s,\ffield^N)} 
  P_{\lspan{Y^\tr}|\lspan{X^\tr}}(V|U)
  \log \frac{\cmat{r}{s}}{P_{\lspan{Y^\tr}|\lspan{X^\tr}}(V|U)} \\
  & = &
  \sum_{s\leq r} p_{\rank(X)\rank(Y)}(r,s) \log {\cmat{r}{s}} + \entropy(\lspan{Y^\tr}|\lspan{X^\tr}),
\end{IEEEeqnarray*}
which proves the first equality in the claim.

For $\bY$ with $\lspan{\bY^\tr}=V$ and $\rank(\bY^\tr)=s$,
we have
\begin{IEEEeqnarray*}{rCl}
  p_{Y^*}(\bY)
  & = &
  \sum_{\bX:\lspan{\bY}\leq \lspan{\bX}} P_{Y^*|X}(\bY|\bX) p_{X}(\bX) \\
  & = &
  \sum_{r} \sum_{U\in\Gr(r,\ffield^M)} \sum_{\bX:\lspan{\bX^\tr} = U,
    \lspan{\bY}\leq \lspan{\bX}}\frac{1}{\cmat{r}{s}} 
  P_{\lspan{Y^\tr}|\lspan{X^\tr}}(V|U) \frac{1}{\cmat{T}{r}} p_{\lspan{X^\tr}}(U) \\
  & = & %
  \sum_{r} \sum_{U\in\Gr(r,\ffield^M)} \frac{1}{\cmat{T}{s}} P_{\lspan{Y^\tr}|\lspan{X^\tr}}(V|U)p_{\lspan{X^\tr}}(U) \\
  & = &
  \frac{1}{\cmat{T}{s}} p_{\lspan{Y^\tr}}(V)
\end{IEEEeqnarray*}
where the third equality follows from
\begin{IEEEeqnarray*}{rCl}
  |\{\bX:\lspan{\bX^\tr} =
    U,\lspan{\bY}\leq \lspan{\bX}\}|
  & = & \sum_{\tilde U\in \Gr(r,\ffield^T):\lspan{\bY}\leq \tilde U} |\{\bX: \lspan{\bX^\tr} =
    U,\lspan{\bX} = \tilde U \}| \\
  & = & \gcos{T}{r}\frac{\cmat{r}{s}}{\cmat{T}{s}} \cmat{r}{r}.
\end{IEEEeqnarray*}
Here, $\{\tilde U\in \Gr(r,\ffield^T):\lspan{\bY}\leq \tilde U\}$
is calculated in Lemma~\ref{lemma:c1}.
Hence,
\begin{IEEEeqnarray*}{rCl}
  \entropy(Y^*)
  & = &
  \sum_{\bY} p_{Y^*}(\bY) \log \frac{1}{p_{Y^*}(\bY)} \\
  & = &
  \sum_s \sum_{V\in \Gr(s,\ffield^N)}
  \sum_{\bY:\lspan{\bY^\tr}=V}
  \frac{p_{\lspan{Y^\tr}}(V)}{\cmat{T}{s}}
  \log \frac{\cmat{T}{s}}{ p_{\lspan{Y^\tr}}(V)} \\
  & = &
  \sum_s \sum_{V\in \Gr(s,\ffield^N)}
  {p_{\lspan{Y^\tr}}(V)}
  \log \frac{\cmat{T}{s}}{ p_{\lspan{Y^\tr}}(V)} \\
  & = &
  \sum_s p_{\rank(Y)}(s) \log \cmat{T}{s} +
  \entropy(\lspan{Y^\tr}).
\end{IEEEeqnarray*}
\end{IEEEproof}

\section{A Technical Lemma}
 \label{sec:prooffood2}

The following lemma gives a lower bound on the difference
  $R(\ffield^M)-R(V)$ for $V\in \Pj(\ffield^M)$.
The intuition behind the bound is that if the input rank is larger,
the output rank also tends to be larger.

\begin{lemma}\label{lemma:food2}
  Consider $\loc(H,T)$ with $T\geq M$. Fix a uniform-given-row-space input.
  For $V\in \Pj(\ffield^M)$ with $\dim(V)=r<\rank^*(H)$,
  \begin{equation*}
    R(\ffield^M) - R(V) > \Theta(T,r,H)\log  q,
  \end{equation*}
  where
  \begin{IEEEeqnarray*}{rCl}
    \Theta(T,r,H) & \triangleq & (T-M) \sum_{k: k>r} \Pr\{\rank(H)
    \geq k\}
    -  r(M-r) + \log_q \cmatt{r}{r}.
  \end{IEEEeqnarray*}
\end{lemma}

\begin{IEEEproof}%
  Let $\tilde U = \ffield^M$.
  Since $V\leq \tilde U$, there exists a full rank
  $M\times M$ matrix
  \begin{equation*}
    \mathbf D = \begin{bmatrix} \mathbf D_0 \\ \mathbf  D_1 \end{bmatrix}
  \end{equation*}
  such that $\lspan{\mathbf{D}^\tr} = \tilde U$ and
  $\lspan{\mathbf{D}_1^\tr} = V$. By
  Lemma~\ref{lemma:cond},
  \begin{align*}
    \sum_{s\geq k} P_{\rank(Y)|\lspan{X^\tr}}(s| V)  & =
    \Pr\{\rank(\mathbf{D_1} H)\geq k\},
  \end{align*}
  and
  \begin{IEEEeqnarray}{rCl}
    P_{\rank(Y)|\lspan{X^\tr}}(s|\tilde U)
    & = & \IEEEnonumber
    \Pr\{\rank(\mathbf D H)= s\} \\
    & = & \label{eq:in1-c}
    \Pr\{\rank(H) = s\}.
  \end{IEEEeqnarray}
  We know $\Pr\{\rank(H) \geq s\} \geq
  \Pr\{\rank(\mathbf{D_1} H)\geq s\}$. So
  \begin{equation}\label{eq:in1-a}
    \sum_{s\geq k} P_{\rank(Y)|\lspan{X^\tr}}(s|\tilde U)
    \geq \sum_{s\geq k} P_{\rank(Y)|\lspan{X^\tr}}(s| V).
  \end{equation}
  Moreover, for $k$ such that $r <k\leq \rank^*(H)$,
  \begin{equation}\label{eq:in1-b}
    \sum_{s: s\geq k} P_{\rank(Y)|\lspan{X^\tr}}(s|V)
    = 0.
  \end{equation}
  Thus,
  \begin{IEEEeqnarray*}{rCl}
  \sum_{s}s (P_{\rank(Y)|\lspan{X^\tr}}(s|\tilde U) -
    P_{\rank(Y)|\lspan{X^\tr}}(s|V))
  & = &
  \sum_k \sum_{s: s\geq k} (P_{\rank(Y)|\lspan{X^\tr}}(s|\tilde U) -
    P_{\rank(Y)|\lspan{X^\tr}}(s|V)) \\
  & \geq & \IEEEyesnumber \label{eq:in1-1}
    \sum_{k: \rank^*(H)\geq k>r}\sum_{s: s\geq k}
    P_{\rank(Y)|\lspan{X^\tr}}(s|\tilde U)  \\
  & \geq & \IEEEyesnumber \label{eq:in1-2}
    \sum_{k: \rank^*(H)\geq k>r} \Pr\{\rank(H) \geq k\} \\
  & \triangleq & \IEEEyesnumber\label{eq:in1}
  \E[H,r],
  \end{IEEEeqnarray*}
  where \eqref{eq:in1-1} is obtained by \eqref{eq:in1-a} and \eqref{eq:in1-b}; \eqref{eq:in1-2} follows from \eqref{eq:in1-c}.

  By the definition of $R(U)$ in \eqref{eq:g},
  \begin{IEEEeqnarray*}{rCl}
    \frac{R(\tilde U) - R(V)}{\log  q}
    & = &
    \sum_{s}
    P_{\rank(Y)|\lspan{X^\tr}}(s|\tilde U) \left((T-M)s+\log_q
      \frac{\cmatt{T}{s}}{\cmatt{M}{s}}\right) 
    - \sum_s P_{\rank(Y)|\lspan{X^\tr}}(s|\tilde
    V) \left((T-r)s+\log_q
      \frac{\cmatt{T}{s}}{\cmatt{r}{s}}\right)  \\
    & = &
    (T-M) \sum_{s}s (P_{\rank(Y)|\lspan{X^\tr}}(s|\tilde
    U) -P_{\rank(Y)|\lspan{X^\tr}}(s| V))  
    - (M-r)\sum_{s}sP_{\rank(Y)|\lspan{X^\tr}}(s| V) \\
    & & \qquad +
    \sum_s P_{\rank(Y)|\lspan{X^\tr}}(s|\tilde U) \log_q
    \frac{\cmatt{T}{s}}{\cmatt{M}{s}} - \sum_s
    P_{\rank(Y)|\lspan{X^\tr}}(s| V) \log_q
    \frac{\cmatt{T}{s}}{\cmatt{r}{s}}  \\
    & > &
    (T-M)\E[H,r] - r(M-r) + \log_q \cmatt{r}{r}, %
  \end{IEEEeqnarray*}
  where the last inequality follows from \eqref{eq:in1},
  \begin{equation*}
    (M-r)\sum_{s}sP_{\rank(Y)|\lspan{X^\tr}}(s|V) \leq r(M-r),
  \end{equation*}
  \begin{equation*}
     \sum_s P_{\rank(Y)|\lspan{X^\tr}}(s|\tilde U) \log_q
    \frac{\cmatt{T}{s}}{\cmatt{M}{s}} \geq 0,
  \end{equation*}
  and
  \begin{equation*}
    \sum_s
    P_{\rank(Y)|\lspan{X^\tr}}(s|V) \log_q
    \frac{\cmatt{T}{s}}{\cmatt{r}{s}}
     < 
    \sum_s
    P_{\rank(Y)|\lspan{X^\tr}}(s|V) \log_q
    \frac{1}{\cmatt{r}{s}} 
     \leq 
    \log_q
    \frac{1}{\cmatt{r}{r}}.
  \end{equation*}
\end{IEEEproof}

\section{Proof of Lemma~\ref{lemma:rsuloc}}
\label{sec:pll}
  We will show that for a row-space-symmetric LOC,
  \begin{equation} \label{eq:pf-dis}
    \mutual(\lspan{X};\lspan{Y}) \leq \locrate(\rank(X);\rank(Y)) + \mutual(\lspan{X};\rank(Y))
  \end{equation}
  where equality holds if and only if $p_{\lspan{Y}}(U)=p_{\lspan{Y}}(U')$
  for all $U$, $U'$ with $\dim(U) = \dim(U')$.  For convenience, we
  call an input distribution $\beta$-type if for any $U\in \ffield^T$,
  there exists $\bX_U$ with $\lspan{\bX_U}=U$ such that
  $p_X(\bX_U)=p_{\lspan{X}}(U)$.  By Lemma~\ref{the:beta}, there must
  exist a $\beta$-type input distribution achieving $C_{\subs}$.  When
  the input distribution is $\beta$-type,
  \begin{IEEEeqnarray*}{rCl}
    \mutual(\lspan{X};\rank(Y))
    &=&
    \mutual(X;\rank(Y)) \\
    &=&
    \mutual(X\lspan{X^\tr};\rank(Y))\\
    &=& \IEEEyesnumber \label{eq:pf-dis2}
    \mutual(\lspan{X^\tr};\rank(Y)),
  \end{IEEEeqnarray*}
  where the first equality is due to the fact that $X$ is
  $\beta$-type, and the last equality follows from the Markov chain
  $X\rightarrow \lspan{X^\tr}\rightarrow \rank(Y)$
  implied by
  Lemma~\ref{lemma:cond}.
  Then, for  a row-space-symmetric LOC,
  \begin{IEEEeqnarray}{rCl}
    C_{\subs}
    & = &  \label{eq:beta-1}
    \max_{p_X:\beta\text{-type}}\mutual(\lspan{X};\lspan{Y}) \\
    & \leq & \label{eq:beta-2}
    \max_{p_X:\beta\text{-type}}\left[\locrate(\rank(X);\rank(Y)) + \mutual(\lspan{X};\rank(Y))\right] \\
    & = & \label{eq:beta-3}
    \max_{p_X:\beta\text{-type}}\left[\locrate(\rank(X);\rank(Y)) + \mutual(\lspan{X^\tr};\rank(Y))\right] \\
    & \leq & \label{eq:beta-4}
    \max_{p_{\lspan{X^\tr}}}\left[\locrate(\rank(X);\rank(Y)) + \mutual(\lspan{X^\tr};\rank(Y))\right],
  \end{IEEEeqnarray}
  where \eqref{eq:beta-1} follows Lemma~\ref{the:beta},
  \eqref{eq:beta-2} is obtained by applying \eqref{eq:pf-dis} for
  row-space-symmetric LOCs, \eqref{eq:beta-3} follows from
  \eqref{eq:pf-dis2}, and \eqref{eq:beta-4} follows that
  $\locrate(\rank(X);\rank(Y))$ and $\mutual(\lspan{X^\tr};\rank(Y))$
  are related to $p_X$ only through $p_{\lspan{X^\tr}}$.

  To prove \eqref{eq:pf-dis}, fix a
  row-space-symmetric LOC.
  Let $\bX$ be an input matrix with rank $r$.
  Consider two subspaces $V'$ and $V$ of
  $\lspan{\bX}$ with dimension $s$. There exists a full rank matrix
  $\Phi$ such that $\Phi V = V'$. Then, by the property of row-space-symmetric LOCs,
  \begin{IEEEeqnarray*}{rCl}
    P_{\lspan{Y}|X}(V'|\bX)
    & = &
    \sum_{\bY:\lspan{\bY}=V'} P_{Y|X}(\bY|\bX) \\
    & = &
    \sum_{\bY:\lspan{\bY}=V} P_{Y|X}(\Phi \bY|\bX) \\
    & = & %
    \sum_{\bY:\lspan{\bY}=V} \frac{1}{\cmat{r}{s}} P_{\lspan{Y^\tr}|\lspan{X^\tr}}(\lspan{\bY^\tr\Phi^\tr}|\lspan{\bX^\tr}) \\
    & = &
    \sum_{\bY:\lspan{\bY}=V} \frac{1}{\cmat{r}{s}} P_{\lspan{Y^\tr}|\lspan{X^\tr}}(\lspan{\bY^\tr}|\lspan{\bX^\tr}) \\
    & = &
    P_{\lspan{Y}|X}(V|\bX).
  \end{IEEEeqnarray*}
  In other words, for all the subspaces $V$ of $\lspan{\bX}$ with the same dimension, $P_{\lspan{Y}|X}(V|\bX)$ are the same.
  Since by Lemma~\ref{lemma:cond}, %
  \begin{equation*}
    P_{\rank(Y)|\lspan{X^\tr}}(s|\lspan{\bX^\tr})
     = 
    P_{\rank(Y)|X}(s|\bX) 
     = 
    \sum_{V\in \Gr(s,\lspan{\bX})} P_{\lspan{Y}|X}(V|\bX),
  \end{equation*}
  we have for any $V\in \Gr(s,\lspan{\bX})$,
  \begin{equation*}
    P_{\lspan{Y}|X}(V|\bX) = \frac{1}{\gcos{r}{s}}P_{\rank(Y)|\lspan{X^\tr}}(s|\lspan{\bX^\tr}).
  \end{equation*}
  Then we have for $V\leq U$ with $\dim(U)=r$ and $\dim(V)=s$,
  \begin{IEEEeqnarray*}{rCl}
    P_{\lspan{Y}|\lspan{X}}(V|U)
    & = &
    \sum_{\bX:\lspan{\bX}=U} P_{\lspan{Y}|X}(V|\bX) P_{X|\lspan{X}}(\bX|U) \\
    & = &
    \sum_{\tilde U\in\Gr(r,\ffield^M)}
    \sum_{\bX:\lspan{\bX}=U,\lspan{\bX^\tr}=\tilde U}
    \frac{1}{\gcos{r}{s}} \\
    & & \times P_{\rank(Y)|\lspan{X^\tr}}(s|\lspan{\bX^\tr}) P_{X|\lspan{X}}(\bX|U) \\
    & = &
    \sum_{\tilde U\in\Gr(r,\ffield^M)}\frac{1}{\gcos{r}{s}}P_{\rank(Y)|\lspan{X^\tr}}(s|\tilde U) P_{\lspan{X^\tr}|\lspan{X}}(\tilde U|U) \\
    & = & \IEEEyesnumber \label{eq:pf-i2d}
    \frac{1}{\gcos{r}{s}} P_{\rank(Y)|\lspan{X}}(s|U).
  \end{IEEEeqnarray*}
  Substituting \eqref{eq:pf-i2d} into the conditional entropy $\entropy(\lspan{Y}|\lspan{X})$, we obtain
  \begin{equation}\label{eq:pf-8s92}
    \entropy(\lspan{Y}|\lspan{X}) = \sum_{r,s}p_{\rank(X)\rank(Y)}(r,s)\log \gcos{r}{s} + \entropy(\rank(Y)|\lspan{X}).
  \end{equation}
  Further, we have
  \begin{IEEEeqnarray*}{rCl}
    \entropy(\lspan{Y})
    & = &
    \entropy(\lspan{Y}\rank(Y))\\
    & = &
    \entropy(\rank(Y)) + \entropy(\lspan{Y}|\rank(Y))\\
    & = &
    \entropy(\rank(Y)) + \sum_s p_{\rank(Y)}(s) \entropy(\lspan{Y}|\rank(Y)=s) \\
    & \leq & \IEEEyesnumber \label{eq:pf-iwi}
     \entropy(\rank(Y)) + \sum_s p_{\rank(Y)}(s)\log \gcos{T}{s}
  \end{IEEEeqnarray*}
  with equality if and only if $$p_{\lspan{Y}}(V) = p_{\rank(Y)}(\dim(V))/\gcos{T}{\dim(V)}$$ for all $V$.
  Therefore, \eqref{eq:pf-dis} is proved by \eqref{eq:pf-8s92} and \eqref{eq:pf-iwi}. 

\section*{Acknowledgement}

We thank Kenneth W. Shum and Raymond W. Yeung for
helpful discussions.

\end{document}